\documentclass[a4paper,11pt]{article}
\usepackage{url}

\usepackage{amsmath,epsfig,subfig,amssymb}
\usepackage{float}
\usepackage{xspace}
\usepackage{algorithm}
\usepackage{algorithmic}
\usepackage{listings} 
\usepackage{amssymb}
\usepackage{graphics}
\usepackage{graphicx}
\usepackage{tikz}
\usetikzlibrary{arrows}
\usetikzlibrary{backgrounds}

\if
  
\else
  
\fi

\usepackage{amsthm}
\newtheorem{lemma}{Lemma}
\newtheorem{definition}{Definition}

\DeclareMathOperator*{\size}{\eta}
\DeclareMathOperator*{\ones}{\omega}
\DeclareMathOperator*{\dom}{\mathbb{D}}
\DeclareMathOperator*{\embdom}{\mathbb{D}_\mathcal{I}}

\DeclareMathOperator*{\gendom}{\mathbb{D}_\mathcal{B}}
\DeclareMathOperator*{\flagdom}{\mathbb{D}_\mathcal{F}}
\DeclareMathOperator*{\genbit}{\mathcal{B}}
\DeclareMathOperator*{\embbit}{\mathcal{I}}
\DeclareMathOperator*{\auxbit}{\mathcal{A}}
\DeclareMathOperator*{\watbit}{\mathcal{E}}
\DeclareMathOperator*{\locbit}{\mathcal{L}}
\DeclareMathOperator*{\complocbit}{\mathcal{L}^C}
\DeclareMathOperator*{\flagbit}{\mathcal{F}}
\DeclareMathOperator*{\compflagbit}{\mathcal{F}^C}

\DeclareMathOperator*{\E}{\mathop{\mathbf{E}^{\size} \/}}
\DeclareMathOperator*{\EP}{\mathop{\mathbf{E}^{\size}_P \/}}
\DeclareMathOperator*{\EPp}{\mathop{\mathbf{E}^{\size}_{P-1} \/}}
\DeclareMathOperator*{\Eones}{\mathop{\mathbf{E}^{\ones} \/}}
\DeclareMathOperator*{\EPones}{\mathop{\mathbf{E}^{\ones}_P \/}}
\begin{document}
%
\title{Estimation of the Embedding Capacity in Pixel-pair based Watermarking Schemes}
\author{Rishabh Iyer, Rushikesh Borse, Ronak Shah and Subhasis Chaudhuri}
\maketitle

\begin{abstract}
Estimation of the Embedding capacity is an important problem
specifically in reversible multi-pass watermarking and is required for analysis before any image can be watermarked. In this paper, we propose an
efficient method for estimating the embedding capacity of a given cover image under
multi-pass embedding, without actually embedding the watermark. We demonstrate this for a class
of reversible watermarking schemes which operate on a disjoint group of pixels, specifically for pixel pairs.
The proposed algorithm iteratively updates the co-occurrence matrix at every stage, to estimate the multi-pass embedding capacity, and is much more efficient vis-a-vis actual watermarking. We also suggest an extremely efficient, pre-computable tree based implementation which is conceptually similar to the co-occurrence based method, but provides the estimates in a single iteration, requiring a complexity akin to that of single pass capacity estimation. We also provide bounds on the embedding capacity.  We finally show how our method can be easily used on a number of watermarking algorithms and specifically evaluate the performance of our algorithms on the benchmark watermarking schemes  of Tian~\cite{tian:reversible} and Coltuc \textit{et al}~\cite{coltuc:RCM}.
\end{abstract}

\section{Introduction} \label{intro}
Reversible Watermarking~\cite{ingemar:digital} is a technique used to preserve the copyright of digital data~(image, audio and video), while at the same time it ensures exact recoverability of the watermark as well as the cover image. This is mainly significant in applications concerning 
military and medical image processing, legal and multimedia archiving of valuable original works,~etc. We briefly describe below some prominent schemes in reversible watermarking.
\subsection{Past work}
There are many algorithms proposed for reversible watermarking, described comprehensively in the survey papers~\cite{caldelli2010reversible, feng:survey}. There are 
four major techniques of embedding watermarks in reversible watermarking schemes, namely:~histogram bin shifting, lossless data compression, expansion and mapping based techniques and prediction based techniques. Histogram bin shifting based techniques~\cite{binshift} suffer from the basic limitation of low embedding capacity, while data compression based
techniques~\cite{celik:lossless} mostly involve computationally expensive algorithms. 
There have been several algorithms proposed and implemented based on transforms on groups of pixels~\cite{alattar:reversible, coltuc:highcapacity, coltuc:RCM, thodi:expansion, tian:reversible, kamastra:reversible, weng2008reversible}, 
because of the basic advantage of high embedding capacity and a modest computation cost. Majority of these techniques operate on a pair of pixels. The first of these was proposed by Tian~\cite{tian:reversible} and subsequently extended by~\cite{alattar:reversible, thodi:expansion, kamastra:reversible}. These techniques are location map based and hence require an additional step of data compression. Recently reversible contrast mapping (RCM) based methods~\cite{borse:capacity, coltuc:highcapacity,coltuc:RCM} have been used to efficiently embed data without using a location map. Prediction based techniques~\cite{thodi:expansion, sachnev:reversible} have also been suggested and they use information from the neighboring pixels to embed information.

\subsection{Multi-pass capacity estimation}\label{motivation}
Loosely the embedding capacity of an image can be described as the size of the largest watermark which can be embedded into that image. Each watermarking technique has a maximum possible embedding capacity over a single pass, and hence often it is necessary to go for multiple passes to embed a much larger watermark into the given image. Recall that multipass embedding involves at every stage, successively embedding the watermark bits into the already watermarked image from the previous stage. Consequently any watermarking application would require an estimation of the number of passes of watermarking possible as well as an analysis of the feasibility of inserting a watermark of specific length into a given image. For this purpose it is necessary
to calculate the embedding capacity beforehand. In practical settings, it may be necessary to find
such estimates repeatedly for different configurations of the watermark and the cover image. Hence it may not be feasible to actually embed the watermark in a given image, and to check if the watermark and the cover image are compatible for
watermarking. Since most watermarking algorithms are quite slow, the possibility of having to check for multiple
iterations of embedding before choosing the right watermark would make the task computationally quite demanding. In particular, these watermarking schemes tend to be quite complex and involve data-compression stages which makes the task of embedding computationally expensive. Hence we require efficient estimation algorithms for the computation of embedding capacity of different
watermarking schemes. The problem of multi-pass embedding capacity estimation has not been studied much in the literature despite this being needed before any cover image or watermark could be selected for embedding. Although Kalker~\cite{kalker:bounds} talks about
capacity bounds based on dirty water codes and Hamming codes, it mainly focuses on capacity bounds for an allowable control distortion. Li \emph{et al}.~\cite{li:reversible} 
also talk about image independent embedding capacity, but the focus lies in finding the minimum possible embedding capacity for any image. Thus there is an urgent need to develop appropriate techniques to compute the embedding capacity in multi-pass watermarking schemes. For a given image, the single pass embedding capacity is simple to estimate and can be directly
computed by considering pixel pairs eligible for watermarking. Multi-pass embedding capacity estimation is however challenging since the subsequent passes depend not just on the cover image but also on the watermarks embedded in the previous iterations. Hence in the rest of the paper, we focus on developing computationally efficient techniques to estimate the embedding capacity in multi-pass embedding.

\section{Framework, Notation and Problem Definition}
\subsection{Framework and Notation}\label{framework}
We show in this paper that for a select class of watermarking algorithms it is indeed possible to provide good estimates of the embedding
capacity over multiple passes. This is a class of transforms which operate on independent groups or blocks of pixels, generally known as the expansion and mapping based algorithms. In this paper, we propose algorithms for pixel-pair based methods, since for majority of these techniques~\cite{borse:capacity, coltuc:highcapacity, thodi:expansion, coltuc:RCM, tian:reversible, kamastra:reversible, weng2008reversible}, the independent group or blocks of pixels are pixel pairs. In other words the entire image is partitioned into disjoint pairs of pixels. We further show that the proposed algorithms can also be extended easily to schemes~\cite{alattar:reversible} which operate on larger groups of pixels (for example pixel triplets and quadruplets). 
\\
\begin{figure}
\begin{centering}
\begin{tikzpicture}[scale=1.5, transform shape, show background rectangle]]
\tikzstyle{block} = [rectangle, fill=blue!30]
\tikzstyle{circ} = [circle, fill=red!30, inner sep=0pt,minimum size=0.5pt]
\node [block](a) at (1, 2) {$\mathbf{P}$};
\node [block](b) at (2, 1) {$\mathbf{F}$};
\node [block](c) at (2,0) {$\mathbf{L}$};
\node [block](d) at (3,2) {$\mathbf{T}$};
\node [block](e) at (3, 1) {$\mathbf{C}$};
\node [block](f) at (3, 0) {$\mathbf{C}$};
\node [circ](g) at (4,1) {\tiny $\odot$};
\node [circ](h) at (4,2) {\tiny $\odot$};
\node [block](i) at (3,3) {$\mathbf{R}$};
\node (x) at (1.5,2){.};
\node (y) at (1.5,1){.};
\node (z) at (1.5,0){.};
\node (w) at (2.5,1){.};
\node (u) at (2.5,0.5){.};
\node (v) at (3.5,0.5){.};
\node (r) at (4,0){};
\draw[->]  (a) -- (1.5,2);
\draw (1.5,2) -- (1.5,1) node[pos=.5,left] {\tiny $\Xi$};
\draw (1.5,1) -- (1.5,0) node[pos=.5,left] {\tiny $\Xi$};
\draw[->] (1.5,2) -- (d) node[pos=.5,above] {\tiny $\Xi$};
\draw[->] (1.5,1) -- (b);
\draw[->] (1.5,0) -- (c);
\draw (b) -- (2.5,1);
\draw[->] (2.5,1) -- (e);
\draw (2.5,1) -- (2.5,0.5);
\draw (2.5,0.5) -- (3.5,0.5);
\draw[->]  (3.5,0.5) -- (g) node[pos=.5,sloped,above] {\tiny $\mathcal{F}$};
\draw [->](e) -- (g) node[pos=.75,above] {\tiny $\mathcal{F}^C$};
\draw (c) -- (f) node[pos=.5,above] {\tiny $\mathcal{L}$};
\draw (f) -- (4,0) node[pos=.5,above] {\tiny $\mathcal{L}^C$};
\draw [->] (4,0) --  (g);
\draw[<-] (d) --  (h) node[pos=.5,above] {\tiny $\mathcal{I}$};
\draw [->] (g) -- (h) node[pos=.5,right] {\tiny $\mathcal{A}$};
\draw [->](4,3) -- (h) node[pos=.5,right] {\tiny $\mathcal{E}$};
\draw [->] (d) --  (i) node[pos=.5,right] {\tiny $\Xi'$};
\draw [->] (i) -- (2,3) node[pos= 1,left] {\tiny Watermarked Image};
\node at (0.2,2) {\tiny Image};
\node at (4,3.25) {\tiny Watermark};
\draw[->]  (0.5, 2) -- (a);
\draw[densely dashed] (1,1.35) -- (4.5,1.35);
\draw[densely dashed] (4.5,1.35) -- (4.5,-0.5) node[pos= 0.85,left] {$\mathbf{A}$};
\draw[densely dashed] (4.5,-0.5) -- (1,-0.5);
\draw[densely dashed] (1,-0.5) -- (1,1.35);
\end{tikzpicture}
\end{centering}
\caption{A block diagram depicting the procedure for this class of watermarking schemes. Here $\mathbf{P}$ denotes the disjoint partitioning block, $\mathbf{F}$ denotes the Flag stream generator, $\mathbf{L}$ represents the location map bitstream generator, $\mathbf{R}$ is the reconstruction block, $\mathbf{C}$ represents the compression method and $\mathbf{T}$ denotes the transform block. }
\label{block}
\end{figure}
\\
We now introduce the notation we will use throughout this paper. Let $\mathbb{D}$ represent the domain of the pixel pairs, i.e $\{\mathbb{D} = [0,L]\times [0,L]\}$, where $L = 255$ for an 8-bit Image. Further let $\xi$ represent a pixel pair $(x,y)$, where $x$ and $y$ refer to the pixel intensities. The procedure involved in pixel-pair based watermarking schemes is depicted in  fig.~\ref{block}. As illustrated in fig.~\ref{block}, the input image is first partitioned through a disjoint partitioning block $\mathbf{P}$, into a set of disjoint pixel pairs, represented as $\Xi = \{\xi_1, \xi_2, \ldots, \xi_N\}$. These pixel pairs are generally adjacent to each other, either horizontally, vertically or diagonally. Let $N$ be the total number of pixel pairs in the image. For ease of notation we will sometimes drop the subscript. We further represent  $\xi' = (x',y')$, as the transformed pixel pair, after watermarking the pixel pair $\xi = (x,y)$.

Recall that in order to prevent overflow and underflow we have the constraints: $0 < x' < L, 0 < y' < L$. Thus not every pixel pair is embeddable. Further in order to control the distortion, some additional restrictions are imposed on the embeddable domain~\cite{coltuc:RCM, tian:reversible, weng2008reversible}. Correspondingly only those pixel pairs $\xi$, such that: $|x-y| < \theta_h$  for some threshold $\theta_h$ are considered for embedding. Thus let the corresponding domain of embeddable pixel pairs be $\embdom \subseteq \mathbb{D}$. Let $\mathbb{B}$ be the set of all possible bitstreams, which are embedded into an image. In order to maintain reversibility of the watermarking algorithms, many times some additional data is required to be embedded along with the watermark bitstream~\cite{feng:survey}. This data is called the auxiliary data stream denoted by $\auxbit \in \mathbb{B}$, and hence the disjoint pixel pairs $\Xi$ are analyzed through an auxiliary data block $\mathbf{A}$ (shown as a densely dashed block in fig. ~\ref{block}), to construct the auxiliary data stream $\auxbit$, needed to be embedded along with the image. Thus within the embeddable pixel pairs, both the watermark and the auxiliary information have to be embedded. There are two types of auxiliary data used in watermarking algorithms belonging to this class. They are the flag bits and the location map. Thus the auxiliary block $\mathbf{A}$ consists of the flag bitstream generator $\mathbf{F}$ and the compressed location map stream generator $\mathbf{L}$. Flag bits are typically required to be embedded along with the watermark to ensure reversibility. In most algorithms these are either in the form of LSB bits of some pixels which need to be stored, or flags containing information regarding some pixel pairs. We represent the domain of pixel pairs which contribute towards the flag bit stream as $\flagdom \subseteq \dom$. In other words for every pixel pair belonging to $\flagdom$, we need to store bits either in the form of flags or LSB bits.  Generally $\flagdom \subset \dom$, as not every pixel pair requires a flag bit to be stored. Let $f_{\xi}$ be the flag contributed by the pixel pair $\xi$. In most algorithms $f_{\xi}$ is a binary bit, and atmost one bit needs to be stored for every pixel pair $\xi$. Further these bits need to be stored only for pixel pairs $\xi \in \flagdom$ and hence for every $\xi \notin \flagdom$, $f_{\xi} =\{\phi \}$. Here $\phi$ represents a nullbit. Let further $\size(f_{\xi})$ represent the number of bits contributed by $\xi$. Hence $\size(f_{\xi}) = 0, \mbox{ if } f_{\xi} = \{\phi\}$.  We can then represent the flag bitstream as $\flagbit = \mathbf{F}(\Xi) =  \{f_{\xi_1},f_{\xi_2}, \ldots, f_{\xi_N}\}$. For those $\xi_j \notin \flagdom$, $f_{\xi_j} = \{ \phi \}$ and they do not contribute towards $\flagbit$. The location map on the other hand is a binary map which stores information for every pixel pair in the image. For example in Tian's scheme~\cite{tian:reversible}, the location map consists of information whether a pixel pair is expandable or not. Let $l_{\xi}$ denote the location map bit corresponding to the pixel pair $\xi$. We denote $\locbit = \mathbf{L}(\Xi) = \{l_{\xi_1}, l_{\xi_2}, \ldots, l_{\xi_N}\}$ as the location map bit stream. Unlike the flagbit stream, the location map is required for every pixel pair in the image (i.e $\size(l_{\xi}) = 1, \forall \xi \in \dom$) and hence it cannot be directly embedded in the form of the auxiliary data since its size is the same as the total number of pixel pairs in the image. Thus it has to be compressed using a compressing method $\mathbf{C}$. We then represent the compressed location map as $\complocbit = \mathbf{C}(\locbit)$. Sometimes the flag bit stream may also be compressed and we represent the compressed flag bit stream as $\compflagbit = \mathbf{C}(\flagbit)$. Hence the auxiliary data stream $\auxbit = \flagbit \odot \complocbit \odot \compflagbit$. Here $\odot$ indicates the concatenation of bitstreams. Hence we have: $\size(\auxbit) = \size(\flagbit) + \size(\complocbit) + \size(\compflagbit)$. Note that one or more of these may be null streams, depending on the watermarking scheme used. As we shall see later, the sizes of the compressed location map and flag bits may depend on the number of ones and zeros in those streams. We assume $\ones(.)$ denotes the number of ones in a bitstream and hence $\ones(\locbit)$ and $\ones(\flagbit)$ denote the number of ones in the location map and flag bitstream respectively.  Let $\mathbb{D}^1_{\mathcal{L}}$ be the region of pixel pairs where the location map bit is $1$. In other words: $\mathbb{D}^1_{\mathcal{L}} = \{\xi \in \dom|  l_{\xi} = 1\}$. Similarly we can define $\mathbb{D}^1_{\mathcal{F}} = \{\xi \in \flagdom| f_{\xi} = 1\}$ as the region of pixel pairs where the flag bit is $1$

We represent the watermark bitstream as $\watbit \in \mathbb{B}$. Correspondingly we can construct the total embeddable bitstream $\embbit = \watbit \odot \auxbit$, with $\size(\embbit) = \size(\watbit) + \size(\auxbit)$. Since we are interested in finding the embedding capacity, we assume that $\watbit$ represents the largest possible bitstream, and correspondingly the embedded bits $\embbit$, will be embedded in every embeddable pixel pair. Thus $\size(\embbit)$ is equal to the total number of embeddable pixel pairs in the image. Let $i_{\xi}$ denote the bit embedded in the pixel pair $\xi$. Further we assume $\size(i_{\xi})$ denotes the number of bits embedded in $\xi$. Given the embeddable bitstream $\embbit$, we can then construct an ordered sequence of bits, to be embedded into the image such that: $\embbit = \{i_{\xi_1}, i_{\xi_2}, \ldots, i_{\xi_N}\}$ with, $i_{\xi_j} = \phi$, if $\xi_j \notin \embdom \mbox{ or } \size(i_{\xi_j}) = 0$. 

Each pixel pair $\xi = (x,y)$ is mapped to another pixel pair $\xi' = (x',y')$ through a transform: $\xi' = T(\xi,i_{\xi})$. Let $\mathbf{T}$ denote the transform block. Let $\xi^0= T(\xi,0)$, $\xi^1 = T(\xi,1)$ and $\xi^{\phi} = T(\xi,\phi)$. Hence $\xi' = \{\xi^0, \xi^1, \xi^{\phi}\}$. Thus we embed the bitstream $\embbit$ into the pixel pairs, $\Xi$, to obtain the set of transformed pixel pairs: $\Xi' = \mathbf{T}(\Xi, \embbit) = \{T(\xi_1,i_{\xi_1}), T(\xi_2,i_{\xi_2}), \ldots, T(\xi_N,i_{\xi_N})\} = \{\xi'_1, \xi'_2, \ldots,\xi'_N\}$. Finally the watermarked image can be obtained by reconstructing the image from the set of pixel pairs $\Xi'$, through a image reconstruction block $\mathbf{R}$. This entire procedure is illustrated in fig.1. The maximum embedding capacity for these algorithms is 0.5 bpp since we can embed atmost a bit in each pixel pair. However the auxiliary data eats up this capacity. Hence $\size(\watbit) = \size(\embbit) - \size(\auxbit)$ represents the size of the largest possible watermark embeddable in the image and is the embedding capacity of the image.

In multipass embedding, the entire embedding stage depicted in fig.~\ref{block}, is repeated at every stage on the watermarked image from the previous stage. Further note that throughout this paper we also assume that at every subsequent pass of embedding the bitstream will be embedded in the same set of disjoint pixel pairs,  which were selected in the first pass of embedding. We represent the number of passes in multi-pass watermarking as $P$.  In multipass embedding, let $\embbit_k, \auxbit_k, \watbit_k, \locbit_k, \flagbit_k,$ etc. refer to the corresponding bitstream of the $k^{th}$ stage of watermarking. In other words, $\embbit_0 = \{i_{\xi_0}, i_{\xi_1}, \cdots, i_{\xi_N}\}, \embbit_1 = \{i_{\xi'_0}, i_{\xi'_1}, \cdots, i_{\xi'_N}\}$ and so on. We represent the final concatenated bitstreams after $P$ passes as $\embbit, \auxbit, \watbit, \flagbit, \locbit,$ etc. For example $\embbit = \embbit_0 \odot \embbit_1 \odot \cdots \odot \embbit_{P-1}$.   Also let $p$ be the fraction of the number of ones in the embedded bitstream to the total size of the bitstream. 

In this paper, we provide a general framework of algorithms applicable to any watermarking scheme fitting in the above mentioned class. We do not confine our analysis to any particular watermarking scheme and try to keep our algorithms as general as possible. However we give some examples below of few watermarking schemes which fit into this framework and briefly describe the embedding regions and the auxiliary data required by them. 

\begin{figure}
\centering
\subfloat[]{
\begin{tikzpicture}
    \node[anchor=south west,inner sep=0] (image) at (0,0) {\includegraphics[height=3.2cm,width=3.2cm]{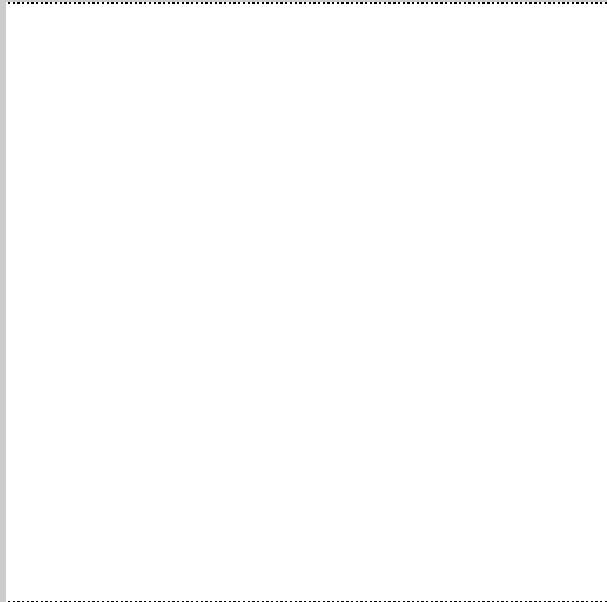}};
    \begin{scope}[x={(image.south east)},y={(image.north west)}]
        \draw[] (0,0) -- (1,0);
	 \draw[] (1,0) -- (1,1);
	 \foreach \n/\x in {0/0, 50/0.19, 100/0.39, 150/0.59, 200/0.78, 250/0.98}
{
   	\draw[] (\x,1) -- (1,1)node[pos= 0,above] {\tiny \n};
}
	 \foreach \n/\x in {250/0.02, 200/0.22, 150/0.41, 100/0.61, 50/0.81, 0/1}
{
   	\draw[] (0,\x) -- (0,1)node[pos= 0,left] {\tiny \n};
}
    \end{scope}
\end{tikzpicture}
\label{tian1}
}
\subfloat[]{
\begin{tikzpicture}
    \node[anchor=south west,inner sep=0] (image) at (0,0) {\includegraphics[height=3.2cm,width=3.2cm]{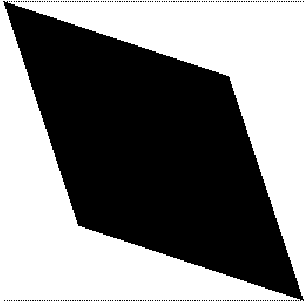}};
    \begin{scope}[x={(image.south east)},y={(image.north west)}]
        \draw[] (0,0) -- (1,0);
	 \draw[] (1,0) -- (1,1);
	 \foreach \n/\x in {0/0, 50/0.19, 100/0.39, 150/0.59, 200/0.78, 250/0.98}
{
   	\draw[] (\x,1) -- (1,1)node[pos= 0,above] {\tiny \n};
}
	 \foreach \n/\x in {250/0.02, 200/0.22, 150/0.41, 100/0.61, 50/0.81, 0/1}
{
   	\draw[] (0,\x) -- (0,1)node[pos= 0,left] {\tiny \n};
}
    \end{scope}
\end{tikzpicture}
\label{tian2}
}\\
\subfloat[]{
\begin{tikzpicture}
    \node[anchor=south west,inner sep=0] (image) at (0,0) {\includegraphics[height=3.2cm,width=3.2cm]{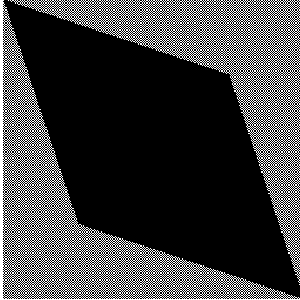}};
    \begin{scope}[x={(image.south east)},y={(image.north west)}]
        \draw[] (0,0) -- (1,0);
	 \draw[] (1,0) -- (1,1);
	 \foreach \n/\x in {0/0, 50/0.19, 100/0.39, 150/0.59, 200/0.78, 250/0.98}
{
   	\draw[] (\x,1) -- (1,1)node[pos= 0,above] {\tiny \n};
}
	 \foreach \n/\x in {250/0.02, 200/0.22, 150/0.41, 100/0.61, 50/0.81, 0/1}
{
   	\draw[] (0,\x) -- (0,1)node[pos= 0,left] {\tiny \n};
}
    \end{scope}
\end{tikzpicture}
\label{tian3}
}
\subfloat[]{
\begin{tikzpicture}
    \node[anchor=south west,inner sep=0] (image) at (0,0) {\includegraphics[height=3.2cm,width=3.2cm]{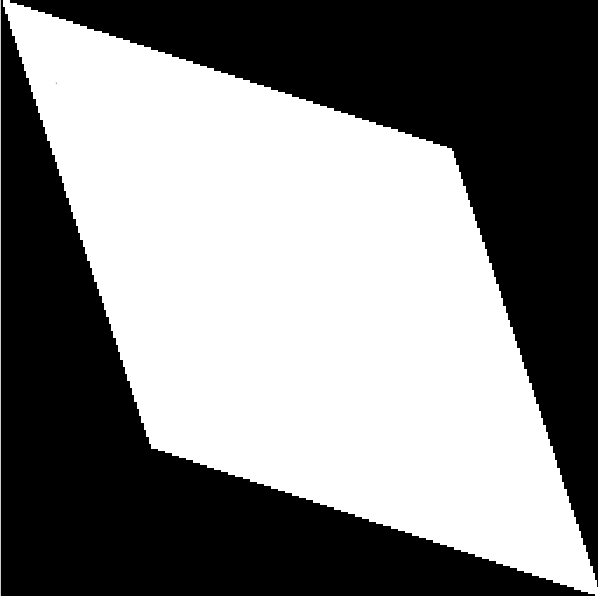}};
    \begin{scope}[x={(image.south east)},y={(image.north west)}]
        \draw[] (0,0) -- (1,0);
	 \draw[] (1,0) -- (1,1);
	 \foreach \n/\x in {0/0, 50/0.19, 100/0.39, 150/0.59, 200/0.78, 250/0.98}
{
   	\draw[] (\x,1) -- (1,1)node[pos= 0,above] {\tiny \n};
}
	 \foreach \n/\x in {250/0.02, 200/0.22, 150/0.41, 100/0.61, 50/0.81, 0/1}
{
   	\draw[] (0,\x) -- (0,1)node[pos= 0,left] {\tiny \n};
}
    \end{scope}
\end{tikzpicture}
\label{tian4}
}
\caption{Illustrations of the regions of $\embdom, \flagdom, \mathbb{D}_{\mathcal{F}}^1,$ and $\mathbb{D}_{\mathcal{L}}^1,$ for the scheme of Tian~\cite{tian:reversible}. The regions are plotted as $255 \times 255$ matrices and the white region represents the corresponding domains. Note that here $\embdom$ is almost the entire region $\dom$. We assume $\theta_h = 255$.}
\label{tian}
\end{figure}
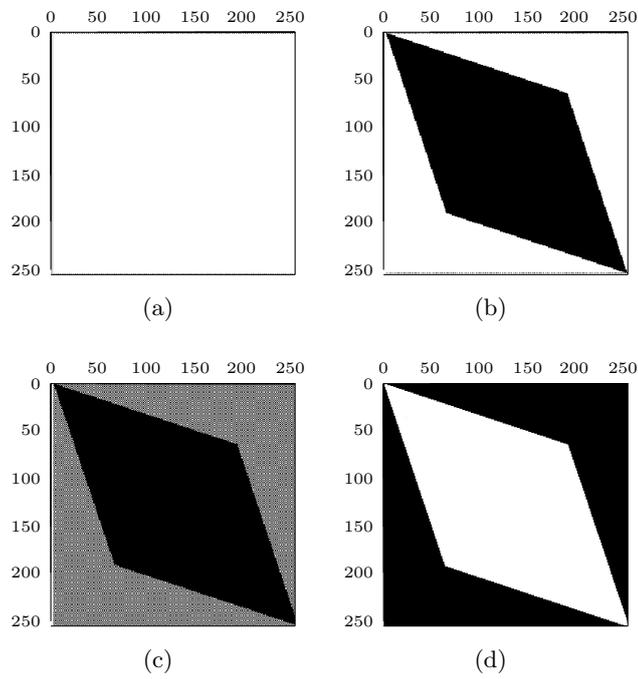

\begin{itemize}
\item Tian~\cite{tian:reversible}: This is the first paper on difference expansion. It embeds a watermark bit into the difference of pixels and requires a location map. Here the set of embeddable pixel pairs is exactly those pixels which are changeable. Further the auxiliary data comprises here of the compressed location map and flag bits in the form of LSB bits. The domains $\embdom, \flagdom, \mathbb{D}^1_{\mathcal{L}}$ and $\mathbb{D}^1_{\mathcal{F}}$ for Tian's scheme are shown in fig.~\ref{tian}. Here we have $f_{\xi} = \mbox{LSB}(|x-y|)$.
\item  Thodi~\textit{et al}~\cite{thodi:expansion}: This is an extension of Tian's algorithm. It uses a combination of histogram bin shifting and difference expansion. The two main algorithms here are difference expansion with histogram shifting using overflow map (DE-HS-OM), and difference expansion with histogram expansion using flag bits (DE-HS-FB). The method is similar to Tian's but they select locations for embedding by defining non overlapping regions in the histogram of the expandable differences. The regions here are similar to those of Tian's and we do not show them separately.
\item Coltuc~\textit{et al}~\cite{coltuc:RCM}: They use a reversible contrast mapping (RCM) method of embedding watermark using a simple integer transform on pairs of pixels. Again the domains $\embdom, \flagdom$ and $\mathbb{D}^1_{\mathcal{F}}$ for Coltuc's scheme are shown in fig.~\ref{coltuc1}, \ref{coltuc2} and \ref{coltuc3}. Further we have here: $f_{\xi} = \mbox{LSB}(x)$.
\item Weng~\textit{et al}~\cite{weng2008reversible}: The concept of invariability of the sum of pixel pairs and pairwise difference adjustment (PDA) is exploited to embed a watermark in a pixel pair. This method requires location map compression. Again the different domains for the scheme of Weng are shown in fig. ~\ref{weng1}. 
\end{itemize}  
\begin{figure}
\centering
\subfloat[]{
\begin{tikzpicture}
    \node[anchor=south west,inner sep=0] (image) at (0,0) {\includegraphics[height=3.2cm,width=3.2cm]{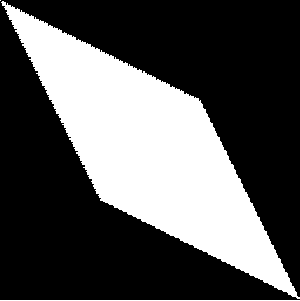}};
    \begin{scope}[x={(image.south east)},y={(image.north west)}]
        \draw[] (0,0) -- (1,0);
	 \draw[] (1,0) -- (1,1);
	 \foreach \n/\x in {0/0, 50/0.19, 100/0.39, 150/0.59, 200/0.78, 250/0.98}
{
   	\draw[] (\x,1) -- (1,1)node[pos= 0,above] {\tiny \n};
}
	 \foreach \n/\x in {250/0.02, 200/0.22, 150/0.41, 100/0.61, 50/0.81, 0/1}
{
   	\draw[] (0,\x) -- (0,1)node[pos= 0,left] {\tiny \n};
}
    \end{scope}
\end{tikzpicture}\label{coltuc1}
}
\subfloat[]{
\begin{tikzpicture}
    \node[anchor=south west,inner sep=0] (image) at (0,0) {\includegraphics[height=3.2cm,width=3.2cm]{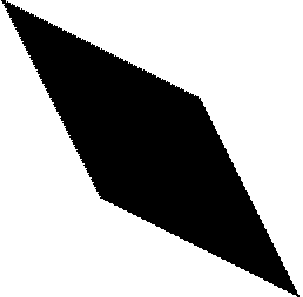}};
    \begin{scope}[x={(image.south east)},y={(image.north west)}]
        \draw[] (0,0) -- (1,0);
	 \draw[] (1,0) -- (1,1);
	 \foreach \n/\x in {0/0, 50/0.19, 100/0.39, 150/0.59, 200/0.78, 250/0.98}
{
   	\draw[] (\x,1) -- (1,1)node[pos= 0,above] {\tiny \n};
}
	 \foreach \n/\x in {250/0.02, 200/0.22, 150/0.41, 100/0.61, 50/0.81, 0/1}
{
   	\draw[] (0,\x) -- (0,1)node[pos= 0,left] {\tiny \n};
}
    \end{scope}
\end{tikzpicture}
\label{coltuc2}
}\\
\subfloat[]{
\begin{tikzpicture}
    \node[anchor=south west,inner sep=0] (image) at (0,0) {\includegraphics[height=3.2cm,width=3.2cm]{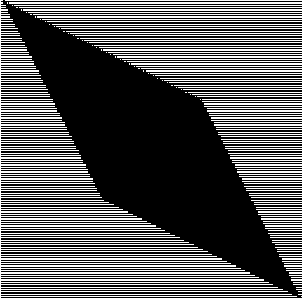}};
    \begin{scope}[x={(image.south east)},y={(image.north west)}]
        \draw[] (0,0) -- (1,0);
	 \draw[] (1,0) -- (1,1);
	 \foreach \n/\x in {0/0, 50/0.19, 100/0.39, 150/0.59, 200/0.78, 250/0.98}
{
   	\draw[] (\x,1) -- (1,1)node[pos= 0,above] {\tiny \n};
}
	 \foreach \n/\x in {250/0.02, 200/0.22, 150/0.41, 100/0.61, 50/0.81, 0/1}
{
   	\draw[] (0,\x) -- (0,1)node[pos= 0,left] {\tiny \n};
}
    \end{scope}
\end{tikzpicture}\label{coltuc3}
}
\subfloat[]{
\begin{tikzpicture}
    \node[anchor=south west,inner sep=0] (image) at (0,0) {\includegraphics[height=3.2cm,width=3.2cm]{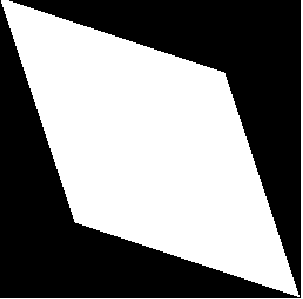}};
    \begin{scope}[x={(image.south east)},y={(image.north west)}]
        \draw[] (0,0) -- (1,0);
	 \draw[] (1,0) -- (1,1);
	 \foreach \n/\x in {0/0, 50/0.19, 100/0.39, 150/0.59, 200/0.78, 250/0.98}
{
   	\draw[] (\x,1) -- (1,1)node[pos= 0,above] {\tiny \n};
}
	 \foreach \n/\x in {250/0.02, 200/0.22, 150/0.41, 100/0.61, 50/0.81, 0/1}
{
   	\draw[] (0,\x) -- (0,1)node[pos= 0,left] {\tiny \n};
}
    \end{scope}
\end{tikzpicture}
\label{weng1}
}
\caption{The regions in above two and the bottom-left figures represent $\embdom, \flagdom$ and $\mathbb{D}_{\mathcal{F}}^1$ respectively for Coltuc's scheme. The bottom left Image represents the region $\embdom$ for the scheme of Weng~\textit{et al}. The regions are plotted as $255 \times 255$ matrices, and the white region represents the corresponding domains.}
\label{mixed}
\end{figure}
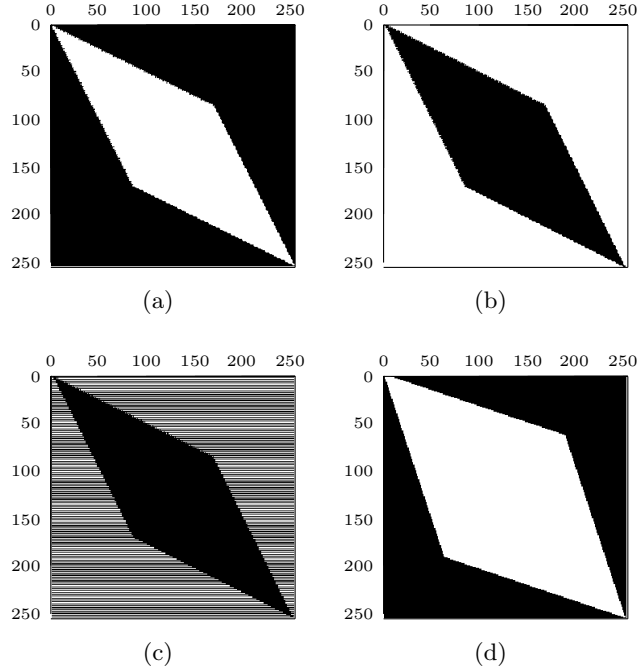 

As it is clear from the above summary, majority of the techniques are location map based and require additional data compression.   The compression of the location map significantly increases the complexity of the algorithms, further emphasizing the need to provide efficient multi-pass embedding capacity estimates as described in section~\ref{intro}. Each watermarking algorithm requires atleast two iterations over the entire image to first identify the embeddable regions and get the auxiliary information and then actually embedding the watermark, in addition to a possibly extra iteration to compress the location map. This is also evident from the watermarking procedure shown in fig.~\ref{block}. In particular the location map compression is a computationally expensive task.
Thus the watermarking schemes have a lot of associated overhead. The RCM based method however does not require any location map, and is comparatively the most efficient algorithm. We however show that our estimation method performs much better that even Coltuc's method computationally, while still providing reasonably precise estimates. 

\subsection{Problem Definition}
We assume that the probability that a given bit in the bitstream is 1 is $p$, and is known. We claim here that due to the block based approach of this class of watermarking schemes, the embedding capacity for a given image and watermark depends only on $p$ and not on the actual bitstream itself. Infact we experimentally verify this in section~\ref{exver}. Hence we pose our problem as providing good estimates of the total embedding capacity $\size(\watbit)$, in multi-pass watermarking schemes given $p$ and the number of passes $P$. Estimating the embedding capacity $\size(\watbit)$ requires estimation of total number of embeddable pixel pairs $\size(\embbit)$, along with the auxiliary data size $\size(\auxbit)$, required to ensure reversibility. 

\subsection{ Our Contributions}
We exploit specific properties of the pixel pair based watermarking schemes to efficiently estimate the multi-pass embedding capacity in computational costs significantly lower than actual watermarking. We first propose a co-occurrence based method which iteratively updates the co-occurrence matrix of the image to effectively estimate both the embedding information and the auxiliary data size required at every pass. Subsequently we propose a pre-computable tree based implementation which concisely represents the multi-pass structure for every pixel pair. We then prove the equivalence between these two methods. We then use these algorithms to estimate the embedding capacity of an image for a given value of $p$. Lastly we also propose methods to estimate bounds on the embedding capacity. Though this is mainly a theoretical paper, we perform a number of experiments and evaluate our algorithms on the watermarking schemes of Tian and Coltuc \textit{et al}. We show that these estimates are reasonably close to the actual capacities of these images.

\section{Embedding capacity estimation}
In this section we present methods to estimate the embedding capacities, given the bitstream distribution $(p)$. The pixel pair tree method though a very fast estimation procedure, requires an offline stage and some additional memory. The co-occurrence matrix method is, however, an iterative method, but is amenable to considering different probabilities at every iteration.

\subsection{Proposed Algorithms}\label{algosdesc}
Let $\genbit \in \mathbb{B}$ represent a general bitstream, to be embedded into an image. It could represent the embedded bitstream ($\embbit$), watermark ($\watbit$), flag bit stream ($\flagbit$) or the location map bit stream ($\locbit$). There are typically two problems of interest related to this. One is estimating the size of the bitstream $\size(\genbit)$, like estimating the number of embeddable pixel pairs $\size(\embbit)$ or the size of the flag bit stream $\size(\flagbit)$. The other is estimating the number of ones in the bitstream $\ones(\genbit)$, which is relevant in estimating the compression of the bitstream, as we shall see later. For example, to estimate the size of the compressed streams $\size(\complocbit)$ or $\size(\compflagbit)$ we would need to estimate the number of ones in the corresponding streams, $\ones(\flagbit) \mbox{ and } \ones(\locbit)$, respectively. Also we assume that $\genbit_k$ denotes the bitstream corresponding to the $k^{th}$ stage of embedding. Let $b_{\xi}$ represent the bits contributed by the pixel pair $\xi$, and $\size(b_{\xi})$ represent the number of bits contributed by the pixel pair $\xi$ towards the bitstream $\genbit$ in a single pass. Since we are concerned mainly with pixel pair based watermarking schemes, $b_{\xi}$ is a bit and $\size(b_{\xi}$ is either 0 or 1, depending whether $\xi$ contributes towards $\genbit$ or not. Again for example,when $\genbit$ is $\embbit, \flagbit \mbox{ and } \locbit$, $b_{\xi}$, is $i_{\xi}, f_{\xi} \mbox{ and } l_{\xi}$ respectively. Further, $\gendom$ represents the region of pixel pairs which contribute towards the bitstream and $\mathbb{D}^1_{\mathcal{B}}$ represents the pixel pairs where $b_{\xi} = 1$. In other words $\gendom = \{\xi \in \dom | b_{\xi} \ne \phi\}$ and $\mathbb{D}^1_{\mathcal{B}} = \{\xi \in \gendom | b_{\xi} = 1\}$.  

\subsubsection{Co-occurrence Matrix based method}\label{co-occuralgo}
\begin{definition} 
(Pair-wise co-occurrence Matrix). We define the pairwise co-occurrence matrix $C$ of size $L \times L$, similar to the conventional co-occurrence matrix~\cite{haralick:cooccur}, as the population (distribution) of co-occurring pixel pairs in an Image. In our context, it represents a count of the number of times a pixel pair occurs in the image. Thus given the disjoint pairs of pixels $\Xi = \{\xi_1, \xi_2, \ldots, \xi_N\}$, we can define it as:    
\begin{equation}\label{co-occurdef}
C(\xi) = \sum_{j = 1}^{N} I(\xi = \xi_j) \mbox{ where  I(.) is the indicator function.}
\end{equation} 
\end{definition}

In this paper all subsequent usage of the term co-occurrence matrix would actually mean a pairwise co-occurrence matrix. We now provide a scheme to iteratively update the co-occurrence matrix and estimate the size of the bitstream at a given stage using the corresponding co-occurrence matrix at that stage. We start with a co-occurrence matrix $C_0$ initially, calculated directly from the cover image. The initial image and its corresponding co-occurrence matrix represent the $0^{th}$ stage. We then iteratively update the co-occurrence matrix at every stage (pass of embedding) using the following scheme: Let the co-occurrence matrix at the $k^{th}$ stage be $C_k$. Then for every pixel pair $\xi \in \embdom$, $p$ fraction of the total number of these pairs of $C_k(\xi)$ will become $\xi^1$, while $1-p$ of these will transform to $\xi^0$ in $C_{k+1}$. For those pixel pairs not embeddable, $\xi$ is transformed to $\xi^{\phi}$. This is elaborated in detail in Algorithm 1.

\begin{algorithm}
\caption{Statistical estimation of the embedding capacity from the co-occurrence matrix for a given probability of watermark $p$ and number of passes $P$.}
\label{algo1}
\algsetup{indent=2em}
\begin{algorithmic}
\STATE Find the pair-wise co-occurrence matrix ${C_0}$ from the given image using Equation~\eqref{co-occurdef}.
\STATE $k = 0$.
\REPEAT
\STATE Set all entries of $C_{k+1}$ to $0$.
\FOR{$\xi \in \embdom$}
\STATE ${C_{k+1}}(\xi^0) \gets C_{k+1}(\xi^0)+(1-p)C_k(\xi).$
\STATE ${C_{k+1}}(\xi^1) \gets C_{k+1}(\xi^1)+pC_k(\xi).$
\ENDFOR
\FOR{$ \xi \notin \embdom$}
\STATE ${C_{k+1}}(\xi^{\phi}) \gets C_{k+1}(\xi^{\phi})+ C_k(\xi).$
\ENDFOR
\STATE $k \gets k+1$
\UNTIL $k < P$                                                           
\end{algorithmic}
\end{algorithm}
The co-occurrence matrix at the $k^{th}$ stage can be used to estimate the embedding capacity for the $(k+1)^{th}$ pass. For example the embedding capacity of the first pass can be estimated from the initial co-occurrence matrix $C_0$. Thus we can write: 
\begin{equation}\label{co-occurstage}
\size(\mathcal{B}_{k}) = \sum_{\xi \in \dom} C_k(\xi) \size(b_{\xi}) = \sum_{\xi \in \gendom} C_k(\xi).
\end{equation}
It is clear that for a $P$ pass watermarking it is sufficient to estimate the co-occurrence matrix upto the $P-1^{th}$ stage. Thus the total size of the bitstream $\size(\genbit)$ can be estimated as:
\begin{equation}\label{co-occurfull}
\size(\genbit) =\sum_{k = 0}^{P-1} \size(\mathcal{B}_{k}) = \sum_{k = 0}^{P-1} \sum_{\xi \in \gendom} C_k(\xi).
\end{equation}
We can similarly find the number of ones in the bitstream $\ones(\mathcal{B}_{k})$, by replacing $\size(b_{\xi})$ by $b_{\xi}$. Thus we have:
\begin{equation}
\ones(\mathcal{B}_{k})\label{co-occuronesstage} = \sum_{\xi \in \dom} C_k(\xi) b_{\xi}= \sum_{\xi \in \mathbb{D}^1_{\mathcal{B}}} C_k(\xi).
\end{equation}
In typical images the co-occurrence matrix is diagonally dominant with only about 10 \% of the entries non-diagonal. Typically the disjoint set of pixel pairs are chosen as the neighbouring pairs of pixels and hence both the pixels are very similar in magnitude and correspondingly the possible set of pixel pairs is actually a small fraction of the total size. Hence one may use the standard sparse matrix representations~\cite{golub1996matrix, pissanetzky1984sparse} to make our algorithms computationally more efficient.

Note that the co-occurrence based method is also an iterative procedure similar to any watermarking scheme. It is however computationally much more efficient than any stage-wise watermarking scheme for the following reasons: 1) The size of the co-occurrence matrix $(256 \times 256)$ is smaller than that of any typical image and hence the computation is much lesser. 2) A large overhead is associated with actual embedding due to the many iterations involved in calculating and embedding the auxiliary data, along with the required compression algorithm. 3) The sparse matrix representation of the co-occurrence matrix further improves the computation. In fact it is many times faster than actual embedding, as we show in the timing analysis in the results.

\subsubsection{Pixel Pair Tree based method}\label{treealgo}
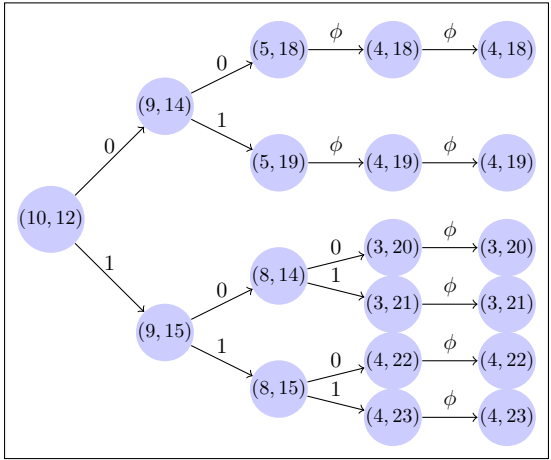
\begin{figure}[t]
\begin{centering}
\begin{tikzpicture}[scale=0.75, transform shape, show background rectangle]]
\tikzstyle{block} = [circle, fill=blue!20, inner sep=0pt,minimum size=0.5pt]
\node [block](a) at (0,0) {\small $(10,12)$};
\node [block](b) at (2,2) {\small  $(9,14)$};
\node [block](c) at (2,-2) {\small  $(9,15)$};
\node [block](d) at (4,3) { \small $(5,18)$};
\node [block](e) at (4,1) { \small $(5,19)$};
\node [block](f) at (4,-1) { \small $(8,14)$};
\node [block](g) at (4,-3) { \small $(8,15)$};
\node [block](h) at (6,3) { \small $(4,18)$};
\node [block](i) at (6,1) { \small$(4,19)$};
\node [block](j) at (6,-0.5) { \small $(3,20)$};
\node [block](k) at (6,-1.5) {\small $(3,21)$};
\node [block](l) at (6,-2.5) { \small $(4,22)$};
\node [block](m) at (6,-3.5) {\small $(4,23)$};
\node [block](n) at (8,3) {\small $(4,18)$};
\node [block](o) at (8,1) { \small$(4,19)$};
\node [block](p) at (8,-0.5) {\small $(3,20)$};
\node [block](q) at (8,-1.5) { \small $(3,21)$};
\node [block](r) at (8,-2.5) { \small $(4,22)$};
\node [block](s) at (8, -3.5) { \small $(4,23)$};
\draw[->] (a) -- (b) node[pos=.5, above] {$0$};
\draw[->] (a) -- (c) node[pos=.5, above] {$1$};
\draw[->] (b) -- (d) node[pos=.5, above] {$0$};
\draw[->] (b) -- (e) node[pos=.5, above] {$1$};
\draw[->] (c) -- (f) node[pos=.5, above] {$0$};
\draw[->] (c) -- (g) node[pos=.5, above] {$1$};
\draw[->] (d) -- (h) node[pos=.5, above] {$\phi$};
\draw[->] (e) -- (i) node[pos=.5, above] {$\phi$};
\draw[->] (f) -- (j) node[pos=.5, above] {$0$};
\draw[->] (f) -- (k) node[pos=.5, above] {$1$};
\draw[->] (g) -- (l) node[pos=.5, above] {$0$};
\draw[->] (g) -- (m) node[pos=.5, above] {$1$};
\draw[->] (h) -- (n) node[pos=.5, above] {$\phi$};
\draw[->] (i) -- (o) node[pos=.5, above] {$\phi$};
\draw[->] (j) -- (p) node[pos=.5, above] {$\phi$};
\draw[->] (k) -- (q) node[pos=.5, above] {$\phi$};
\draw[->] (l) -- (r) node[pos=.5, above] {$\phi$};
\draw[->] (m) -- (s) node[pos=.5, above] {$\phi$};

\end{tikzpicture}
\end{centering}
\caption{ An illustration of the pixel-pair tree for the pixel pair $(10,12)$ for Coltuc's method~\cite{coltuc:RCM}. A similar tree can be constructed for any block based watermarking scheme.}
\label{treealgofig}
\end{figure}	

\begin{definition}
(Pixel-pair tree).  The pixel-pair tree for a pair $\xi$ is defined as a tree which starts 
with the pixel pair $\xi$ and traces a specific path based on the subsequent embedded bits as this pixel pair evolves and represents all feasible paths of watermarking for a given number of passes. 
\end{definition}
A pixel pair tree for a pixel pair $(10,12)$ embedded using the watermarking scheme of Coltuc~\cite{coltuc:RCM}, is shown in fig.~\ref{treealgofig}. This tree concisely represents the entire life cycle of a particular pixel pair upto $P$ passes. We begin by first defining the notation we will use for the pixel pair tree. Let $\mathbf{S}_{\xi}$ represent the set of all paths for a given pixel pair. Note that we are assuming here that this is a $P-$stage watermarking. For example in fig.~\ref{treealgofig} the possible paths for the pair (10,12) upto $P = 4$ passes are \{(10,12), (9,14), (5,18), (4,18), (4,18)\}, \{(10,12), (9,14), (5,19), (4,19), ((4,19)\} and so on. Let $s$ denote a particular path in $\mathbf{S}_{\xi}$ and let $s[k]$ denote the pixel pair at the $k^{th}$ stage of embedding in the path $s$.  Again, consider the path $s = \{(10,12), (9,14), (5,18), (4,18), (4,18)\}$, then we have $s[0]$ as the pixel pair (10,12), $s[1]$ as (9,14), $s[2]$ as (5,18) and so on. Similarly, ${i}_{s[0]}$, ${i}_{s[1]}$ and ${i}_{s[2]}$ are 0,0 and $\phi$ respectively and  $\size({i}_{s[0]})$, $\size({i}_{s[1]})$ and $\size({i}_{s[2]})$ are 1,1 and 0 respectively. Further let $s_j$ denote a path in the pixel pair tree $\mathbf{S}_{\xi_j}$ (the pixel pair tree for the pair $\xi_j$). Then we define $\mathbf{s} = \{s_1, s_2, \cdots, s_N\}$ as a particular path configuration chosen for the pixel pairs $\Xi = \{\xi_1, \xi_2, \cdots, \xi_N\}$ and let $\mathbb{S}$ represent the set of all possible path configurations $\mathbf{s}$. Let $p_s$ denote the path specific probability of the path $s$. It represents the probability that the pixel pair will evolve through the specific path $s$. Then clearly, we can write $p_s = \overset{P-1}{\underset{k = 0}{\prod}} p_{s[k]}$, where $p_{s[k]}$ denotes the probability of transition from $s[k]$ to $s[k+1]$ given whether $s[k]$ is embeddable. Thus 
\begin{equation}\label{treeprobdef}
p_{s[k]} =
\left\{
	\begin{array}{lll}
		p, & \mbox{if } i_{s[k]} = 1 \mbox{ \& }s[k] \in {\embdom} \\
		1-p, & \mbox{ if } i_{s[k]} = 0 \mbox{ \& }s[k] \in {\embdom} \\
		1, & \mbox{ if }s[k] \notin \embdom
	\end{array}
\right.
\end{equation}
Let $\size(\mathbf{s},\genbit)$ represent the size of the bitstream obtained through a specific path configuration $\mathbf{s} = \{s_1, s_2, \cdots, s_N\}$. Further let $\size(\xi, s, \genbit)$ denote the total number of bits contributed by a specific pixel pair $\xi$ through the path $s$. In other words $\size(\xi, s, \genbit) = \overset{P-1}{\underset{k = 0}{\sum}} \size(b_{s[k]})$ and $\size(\mathbf{s},\genbit) = \overset{N}{\underset{j = 1}{\sum}} \size(\xi_j, s_j, \genbit)$. The estimate of the length of the bitstream $\genbit$ can then be given as $\size(\genbit) =  \mathbf{E}_{\mathbf{s} \in \mathbb{S}} (\size(\mathbf{s},\genbit)) $, where $ \mathbf{E}_{\mathbf{s} \in \mathbb{S}} (\size(\mathbf{s},\genbit)) $ represents the total expected size of $\genbit$ by considering every possible path configuration $s_1 \in \mathbf{S}_{\xi_1}, s_2 \in \mathbf{S}_{\xi_2}, \cdots, s_N \in \mathbf{S}_{\xi_N}$, which we can also write as $\mathbf{s} \in \mathbb{S}$. We then write this expectations $\mathbf{E}_{\mathbf{s} \in \mathbb{S}} (\size(\mathbf{s},\genbit))$, as the sum of the expected number of bits ($\mathbf{E}_{s \in \mathbf{S}_{\xi}} (\size(\xi, s, \genbit))$) contributed by every pixel pair $\xi$.  Since we have considered here a $P$ stage watermarking, for simplicity of notation, we represent this expectation as $\EP(\xi,\genbit)$. Thus we can write:
\begin{eqnarray}\label{treeestimatefull}
\size(\genbit) &=& \mathbf{E}_{\mathbf{s} \in \mathbb{S}} (\size(\mathbf{s},\genbit)) = \sum_{j=1}^N \mathbf{E}_{s_j \in \mathbf{S}_{\xi_j}} (\size(\xi_j, s_j, \genbit) ) \nonumber\\
			&=& \sum_{j=1}^N \EP(\xi_j,\genbit) \nonumber = \sum_{j=1}^N \sum_{s_j \in \mathbf{S}_{\xi_j}} p_{s_j} \size(\xi_j, s_j, \genbit)  \nonumber \\
			&=&  \sum_{j=1}^N\: \sum_{s_j \in \mathbf{S}_{\xi_j}} \biggl(~\prod_{k=0}^{P-1} {p_{s_j[k]}} \biggl) \sum_{k=0}^{P-1} {\size({b}_{s_j[k]})}
\end{eqnarray} 

In a similar manner we can also find the size of the bitstream $\size(\genbit_k)$ as an expectation over every possible path in the pixel pair tree. We use the notations similar to our above derivations, with each quantity now representing its stage-wise equivalent. However note that here $\size(\xi, s, \mathcal{B}_k) ) = \size(b_{s[k]})$. Further since we are interested in estimating the size of $\genbit_k$, the pixel pair tree extends only till the $k^{th}$ stage and hence $p_s = \overset{k}{\underset{m = 0}{\prod}} p_{s[m]}$. Again for convenience we use $\E(\xi,\mathcal{B}_k)$ instead of $\mathbf{E}_{s \in \mathbf{S}_{\xi}} (\size(\xi, s, \mathcal{B}_k) )$ and we can write:
\begin{eqnarray}\label{treeestimatestage}
\size(\mathcal{B}_k) &=& \mathbf{E}_{\mathbf{s} \in \mathbb{S}} (\size(\mathbf{s},\mathcal{B}_k)) = \sum_{j=1}^N \mathbf{E}_{s_j \in \mathbf{S}_{\xi_j}} (\size(\xi_j, s_j, \mathcal{B}_k) ) \nonumber \\ 
			&=& \sum_{j = 1}^N \E(\xi_j,\mathcal{B}_k) 	= \sum_{j=1}^N \sum_{s_j \in \mathbf{S}_{\xi_j}} \size(\xi_j, s_j, \mathcal{B}_k) p_{s_j} \nonumber \\
			&=& \sum_{j=1}^N \sum_{s_j \in \mathbf{S}_{\xi_j}}  {\size({b}_{s_j[k]})} \biggl(~\prod_{m=0}^{k} {p_{s_j[m]}} \biggl) 
\end{eqnarray} 
Thus we summarize the expressions for $\EP(\xi,\genbit)$ and $\E(\xi,\mathcal{B}_k)$ as:
\begin{eqnarray}\label{totalex}
\EP(\xi,\genbit) = \: \sum_{s \in \mathbf{S}_{\xi}} \biggl(~\prod_{k=0}^{P-1} {p_{s[k]}} \biggl) \sum_{k=0}^{P-1} {\size({b}_{s[k]})} \\
\label{stageex} \E(\xi,\mathcal{B}_k) = \sum_{s \in \mathbf{S}_{\xi}}  {\size({b}_{s[k]})} \biggl(~\prod_{m=0}^{k} {p_{s[m]}} \biggl) 
\end{eqnarray}
We can then find the number of ones in the bitstreams $\mathcal{B}_k$ and $\genbit$. This problem is very similar to that of finding the size of the bitstreams, since now instead of summing over $\size(b_{s[k]}$ we need to sum over the actual bits contributed $b_{s[k]}$. Hence we do not again reformulate the problem, but use the same procedure above, by just replacing $\size(b_{s[k]})$ by $b_{s[k]}$ and give the final results as:
\begin{eqnarray}
\ones(\genbit) = \sum_{j=1}^N \EPones(\xi_j,\genbit) = \sum_{j=1}^N \: \sum_{s_j \in \mathbf{S}_{\xi_j}} \prod_{k=0}^{P-1} {p_{s_j[k]}} \sum_{k=0}^{P-1} {{b}_{s_j[k]}}  \nonumber\\
\ones(\mathcal{B}_k) = \sum_{j=1}^N \Eones(\xi_j,\mathcal{B}_k) = \sum_{j=1}^N  \: \sum_{s_j \in \mathbf{S}_{\xi_j}}  {{b}_{s_j[k]}} \prod_{m=0}^{k} {p_{s_j[m]}}
\end{eqnarray} 
As evident from above, the computation of the number of ones of a bitstream is exactly the same as that of computing the size of that bitstream  except for replacing $\size({b}_{s[k]})$ by ${b}_{s[k]}$. In the rest of the analysis we provide some important theorems and properties related to the total size of the bitstreams as well as the size of the bitstreams at every stage. These properties however also hold for the number of ones in the these bitstreams.
\begin{itemize}
\item It is easy to show that: $\EP(\xi,\genbit) = \sum_{k = 0}^{P-1} \E(\xi,\mathcal{B}_{k})$. 
\item The formulations above are in terms of the image pixel pair (equations~\eqref{treeestimatestage} and \eqref{treeestimatefull}), and we call them the image pixel-pair based formulation. We can alternatively also reformulate them in terms of the co-occurrence matrix. From the definition of the co-occurrence matrix, it is clear that we can write-
\begin{equation} \size(\mathcal{B}_k) =  \sum_{\xi \in \dom} C_0(\xi) \E(\xi,\mathcal{B}_{k}) \end{equation}
\begin{equation} \size(\genbit) = \sum_{\xi \in \dom} C_0(\xi) \EP(\xi,\genbit) \end{equation}
\item The co-occurrence based algorithm and the tree based algorithm provide exactly the same estimates. In other words, the estimates provided by equations~\eqref{co-occurstage} and~\eqref{co-occurfull} are exactly the same as the ones provided in equation~\eqref{treeestimatestage} and~\eqref{treeestimatefull}. We prove this formally in theorem-~\ref{thm3} (provided in Appendix-A). 
\item The formulation of $\EP(\xi,\genbit)$ and $\E(\xi,\genbit_k)$ requires an exhaustive search over every possible path making them exponential.  Hence we need to provide a more efficient estimation mechanism. We observe that there exists a recursive relationship, which we describe in the following theorem.  For simplicity we represent $\size(b_{s[0]})$ as $\size(b_{\xi})$ since $s[0] = \xi$. 
 \end{itemize}
\newtheorem{theorem}{Theorem}
\begin{theorem}
Equation~\eqref{stageex} can be reformulated in a recursive manner (for $k > 0$) as: 
\begin{equation}\label{stageexrec}
\E(\xi,\mathcal{B}_k) =
\left\{
	\begin{array}{llll}
		(1-p)\E(\xi^0, \mathcal{B}_{k-1}) + p \E(\xi^1,\mathcal{B}_{k-1}) , & \mbox{if } \xi \in \embdom \\
		\E(\xi^{\phi},\mathcal{B}_{k-1}), & \mbox{if } \xi \notin \embdom
	\end{array}
\right.
\end{equation}
with the base case of: $\E(\xi,\mathcal{B}_0) = \size(b_{\xi})$. Equation~\eqref{totalex} can also be reformulated (for $P > 1$) as:
\begin{equation}\label{totalexrec}
\EP(\xi,\genbit) =
\left\{
	\begin{array}{llll}
		\size({b}_{\xi}) + p \EPp(\xi^1,\genbit) + (1-p)\EPp(\xi^0,\genbit), & \mbox{if } \xi \in \embdom\\
		\size({b}_{\xi}) + \EPp(\xi^{\phi}, \genbit), & \mbox{if } \xi \notin \embdom
	\end{array}
\right.
\end{equation}
with the base case: $\E_1(\xi,\genbit) = \size(b_{\xi})$. 
\end{theorem} 

\begin{proof}
The  base cases can be derived from equations~\eqref{stageex} and ~\eqref{totalex}. We first prove it for the case of the stage-wise expectation. Then we consider for $k > 0$, the case where $\xi$ is embeddable. We define $\mathbf{S}_{\xi^0}$ and $\mathbf{S}_{\xi^1}$ as the set of possible paths for the pixel pairs $\xi^0$ and $\xi^1$, respectively in $P-1$ passes. Also let $s^0, s^1$ represent a possible path in $\mathbf{S}_{\xi^0}$ and $\mathbf{S}_{\xi^1}$, respectively. The main idea here is that we break up the set of paths $\mathbf{S}_{\xi}$ into 2 groups, one of which starts with $\xi$ and consists of the paths in $\mathbf{S}_{\xi^0}$, while the other also starts with $\xi$ but consists of those in $\mathbf{S}_{\xi^1}$. Then we can rewrite equation~\eqref{stageex} as:
\begin{eqnarray}\label{tempmax}
 \E(\xi,\mathcal{B}_k) = \sum_{s^1 \in \mathbf{S}_{\xi^1}} p \prod_{m=0}^{k-1} {p_{s^1 [m]}}  \size(b_{s^1 [k-1]}) + \sum_{s^0 \in \mathbf{S}_{\xi^0}} (1-p) \prod_{m=0}^{k-1} {p_{s^0 [m]}}  \size(b_{s^0 [k-1]})
\end{eqnarray}
Note that the set of pixel pairs $\forall s \in \mathbf{S}_{\xi}, s[k]$ are captured fully by the pixel pairs $\forall s^1 \in \mathbf{S}_{\xi^1}, s^1 [k-1]$ and $\forall s^0 \in \mathbf{S}_{\xi^0}, s^0 [k-1]$. Also recognize that within each summation, we have $\E(\xi^0,\genbit_{k-1})$ and $\E(\xi^1,\genbit_{k-1})$ embedded and hence we can rewrite equation~\eqref{tempmax} as:
\begin{equation}
 \E(\xi,\mathcal{B}_k) = p \E(\xi^0,\mathcal{B}_{k-1}) + (1-p)  \E(\xi^1,\mathcal{B}_{k-1})  \nonumber
\end{equation}
We can similarly prove this for the case when $\xi$ is not embeddable.

We now consider the case for the total expectation of a given pixel pair. Again we use the same idea as above and rewrite equation~\eqref{totalex} as:
\begin{eqnarray}
 \EP(\xi,\genbit) = \sum_{s^1 \in \mathbf{S}_{\xi^1}} p \prod_{k=0}^{P-2} {p_{s^1 [k]}} (\sum_{k=0}^{P-2} {\size(b_{s^1 [k]}) + \size(b_{\xi}))} +\sum_{s^0 \in \mathbf{S}_{\xi^0}} (1-p) \prod_{k=0}^{P-2} {p_{s^0 [k]}}  (\sum_{k=0}^{P-2} \size({b_{s^0 [k]}) + \size(b_{\xi}))}  \nonumber
\end{eqnarray}
We observe that $\size(b_{\xi})$ occurs in every term. Again we recognize that within each of the summations are $\EPp(\xi^0,\genbit)$ and $\EPp(\xi^1,\genbit)$. Thus we can write:
\begin{eqnarray}\label{312}
 \EP(\xi,\genbit) = (1-p) \EPp(\xi^0,\genbit) + p  \EPp(\xi^1,\genbit) + \size(b_{\xi}) \biggl( p \sum_{s^1 \in \mathbf{S}_{\xi^1}} \prod_{k=0}^{P-2} {p_{s^1 [k]}} +  (1 - p) \sum_{s^0 \in \mathbf{S}_{\xi^0}}  \prod_{k=0}^{P-2} {p_{s^0 [k]}} \biggl)
\end{eqnarray}
Since the paths in $\mathbf{S}_{\xi^0}$ and $\mathbf{S}_{\xi^1}$ form a complete subtree, $\sum_{s^1 \in \mathbf{S}_{\xi^1}} \prod_{k=0}^{P-2} {p_{s^1 [k]}} = 1$ and  \\ $\sum_{s^0 \in \mathbf{S}_{\xi^0}}  \prod_{k=0}^{P-2} {p_{s^0 [k]}} = 1$. Thus we observe that equation~\eqref{312} transforms to equation~\eqref{totalexrec} for the case when $\xi$ is embeddable. We can similarly handle the case when $\xi$ is not embeddable. 
\end{proof}

Note that we can similarly reformulate $\Eones(\xi,\mathcal{B}_k)$ and $\EPones(\xi,\genbit)$  recursively and the expression is similar to~\eqref{stageexrec} and \eqref{totalexrec} except for replacing $\size({b}_{\xi})$ by ${b}_{\xi}$.

It is clear from equation~\eqref{totalexrec} that $\EP(\xi,\genbit)$ or the total expected number of bits for a given pixel pair is a polynomial in $p$. Hence in order to make the offline stage independent of $p$, we compute the  polynomial coefficients of $\EP(\xi,\genbit)$ for every pixel pair and store them. Thus we can directly use these stored coefficients to compute $\EP(\xi,\genbit)$ for any given $p$ without having to recompute them every time. In addition the set of coefficients is typically iteratively computed for many values of $P$. Hence we can use a simple 'memoization'~\cite{michie1968memo} while computing the total expectation for a given pixel pair $\EP$ and using the previously computed values of $\EPp$, in equation~\eqref{totalexrec}. Thus the tree formulation of equation~\eqref{totalex} can be computed in linear time using these simple tricks. The same idea can be extended for computing the stage-wise expectation.

Thus there are two stages of computation, i.e the online stage and offline. The offline stage is image independent and we compute the coefficients of the polynomials, 
described in equation \eqref{totalexrec} and equation \eqref{stageexrec}. These can be iteratively computed for various values of $P$ in linear time. The online stage consists of just a single iteration to use the values of $\EP(\xi,\genbit)$ or $\E(\xi,\genbit_k)$ depending on what has to be estimated. 

Thus this method provides an extremely efficient implementation to estimate both the size and the number of ones in a bitstream. Importantly it can compute this in a single iteration, although the embedding is multi-pass. Correspondingly, we can find the multi-pass embedding capacity, in a time complexity comparable to that of single pass estimation. Thus it is more efficient than the co-occurrence based estimation framework. The only drawback is that it requires an offline stage and additional memory to store the total or the stage-wise expectation for every possible pixel pair. 

\subsection{Compressed bit-streams}\label{compstream}
Compressed bitstreams sometimes occur as auxiliary data, either as a compressed location map or compressed flag bits. Correspondingly it is necessary to estimate the size of these compressed streams. Let $\genbit^C$ be the bitstream obtained by compressing $\genbit$. Further let $Cf(.)$ represent the compression factor (a number between 0 and 1) for a bitstream and a given compression scheme. We can then find the size of the compressed bitstream as:
\begin{equation}\label{compmapgen}
\size(\mathcal{B}^C)=  \sum_{k = 0} ^{P-1} \size(\mathcal{B}_k) Cf(\mathcal{B}_k)
\end{equation}
However, we cannot find the exact bitstreams at every pass $\genbit_k$, just using the bit probability, since they depend on the exact watermark sequence. Further we would have to actually embed the watermark in order to find the transformed pixel pairs and hence our statistical framework cannot be extended to this. However the Kolmogorov complexity of a binary bitstream of length $n$, with $r$ number of ones has been shown~\cite{cover1991elements} to be bounded by $n H_0(\frac{r}{n})$, where $H_0(.)$ refers to the entropy~\cite{shannon271948} of a bitstream. Correspondingly the compression of these bitstreams, can be estimated by just finding the number of ones $\ones(\genbit_k)$ in these bitstreams at every stage. These estimates though approximate are still quite good and useful in practice. In the last section we proposed algorithms to estimate $\ones(\genbit_k)$. Hence we can estimate the size of the compressed bitstream at every stage $\size(\genbit_k^C)$ as well as the total compressed bitstream size $\size(\genbit^C)$ as:
\begin{equation}
\size(\mathcal{B}_k^C) = \size(\mathcal{B}_k) H_0(\frac{\ones(\mathcal{B}_k)}{\size(\mathcal{B}_k)})
\end{equation}
\begin{equation}
\size(\mathcal{B}^C) =  \sum_{k = 0} ^ {P-1} \size(\mathcal{B}_k) H_0(\frac{\ones(\mathcal{B}_k)}{\size(\mathcal{B}_k)}).
\end{equation}

\subsection{Embedding Capacity Estimation}~\label{estnn}
Using the tools we have developed in the earlier sections, we are now in a position to estimate the embedding capacity. In particular, we use the co-occurrence or the tree based algorithm to find the total number of embeddable pixel pairs ($\size(\embbit)$ and the auxiliary data size. Hence we can find the total and the stage-wise embedding capacity as:
\begin{equation}
\size(\mathcal{E}_k) = \size(\mathcal{I}_k) - \size(\mathcal{F}_k) - \size(\mathcal{F}^C_k) - \size(\mathcal{L}^C_k)
\end{equation}
\begin{equation}
\size(\watbit) = \size(\embbit) - \size(\flagbit) - \size(\compflagbit) - \size(\complocbit)
\end{equation}
Note that generally all these types of auxiliary data will not occur together in any watermarking scheme and hence one or more of these will be null streams. We give below estimates of each of $\size(\mathcal{I}_k), \size(\mathcal{F}_k), \size(\mathcal{F}^C_k) $ and $\size(\mathcal{L}^C_k)$ in terms of the co-occurrence and tree based algorithms.
\begin{equation}
\size(\mathcal{I}_k) = \sum_{\xi \in \embdom} C_k(\xi) = \sum_{j = 1} ^ N \E(\xi_j, \mathcal{I}_k)
\end{equation}
\begin{equation}
\size(\mathcal{F}_k) = \sum_{\xi \in \flagdom} C_k(\xi) = \sum_{j = 1} ^ N \E(\xi_j, \mathcal{F}_k)
\end{equation}
Further we can find the estimates of the number of ones in $\flagbit$ and $\locbit$ as:
\begin{equation}
\ones(\mathcal{F}_k) = \sum_{\xi \in \mathbb{D}^1_{\mathcal{F}}} C_k(\xi) = \sum_{j = 1} ^ N \Eones(\xi_j, \mathcal{F}_k)
\end{equation}
\begin{equation}
\ones(\mathcal{L}_k) = \sum_{\xi \in \mathbb{D}^1_{\mathcal{L}}} C_k(\xi) = \sum_{j = 1} ^ N \Eones(\xi_j, \mathcal{L}_k)
\end{equation}
Finally we can find  $\size(\mathcal{F}^C_k) $ and $\size(\mathcal{L}^C_k)$ as:
\begin{equation}
\size(\mathcal{F}^C_k) = \size(\mathcal{F}_k) H_0(\frac{\ones(\mathcal{F}_k)}{\size(\mathcal{F}_k)}), \,\, \size(\mathcal{L}^C_k) = N H_0(\frac{\ones(\mathcal{L}_k)}{N})
\end{equation}
We have $\size(\mathcal{L}_k) = N$, since the size of the location map is exactly that of the number of pixel pairs.
We can then find $\size(\mathcal{I}), \size(\mathcal{F}), \size(\mathcal{F}^C) $ and $\size(\mathcal{L}^C)$ by summing each of the above obtained expressions from $k = 1 \mbox{ to } P-1$. The tree based estimates for $\size(\mathcal{I}) \mbox{ and } \size(\mathcal{F})$ can directly be obtained however, by just replacing the stage-wise expectation $\E$ by the total expectation $\EP$.

\subsection{Change in probability of ones due to auxiliary data stream}\label{cap}
One subtle point worth mentioning here is that  the embedded bitstream contains both the auxiliary data as well as the watermark. We have until now represented $p$ as the probability of the bit '$1$' in the embedded bit stream. Let $p_W$ represent the probability of the bit '$1$' in the watermark. Then $p \approx p_W$ since the watermark is the major component of the embedded bitstream atleast for the initial passes. However, since a few watermarking schemes (for example Coltuc's method) significantly depend on the probability of the embedded bits, it may be necessary to find at every stage the probability of the bit '$1$' in the embedded bitstream.  For this we need to estimate the probability of '$1$' in the auxiliary bitstream, which we denote by $p_A$. In order to estimate $p_A$, we need to estimate the number of ones in the auxiliary data stream. In the case where the auxiliary data is represented as compressed bit streams, we can assume that the compressed bitstream will be random, and the number of ones and zeros are the same. Hence we can assume the probability of the bits contributed by these compressed bitstreams is $0.5$. The number of ones in the flag bit stream can be computed, as $\ones(\flagbit_k)$, and hence the probability $p_A$ can be found as:
\begin{equation}\label{stageauxprob}
p_{A_k} = \frac{\ones(\flagbit_k) + 0.5(\size(\complocbit_k) + \size(\compflagbit_k))} {\size(\flagbit_k) + \size(\complocbit_k) + \size(\compflagbit_k)}
\end{equation}
We can then easily modify algorithm-1, and use $p = p_k$ at every stage of the algorithm, where $p_k$ can be defined as:
\begin{equation}\label{weighprob}
p_k = \frac{(\size(\embbit_k) - \size(\auxbit_k)).p_W + \size(\auxbit_k).p_{A_k}}{\size(\embbit_k)}
\end{equation}
Here $\auxbit$ contains combinations from the flag bit stream, compressed location map and compressed flag stream. We call this method the Co-occurrence based adaptive probability (CAP) algorithm. Using the co-occurrence matrix at the $k^{th}$ stage $C_k$, we first estimate $\size(\embbit_k)$ and $\size(\auxbit_k)$ at every iteration using techniques discussed in section~\ref{estnn}. These estimates are then used in equations~\eqref{stageauxprob}, \eqref{weighprob} to update the probabilities $p_k$. Finally using the updated probabilities $p_k$ and methods given in algorithm-I, we estimate $C_{k+1}$. Thus we can modify the iterative co-occurrence based method to consider weighted probabilities at every iteration. The tree based implementation, however, inherently considers a single probability and is not amenable to a change in probability at every iteration and hence cannot be used to provide accurate estimates. However this change is necessary only for those schemes which significantly depend on the probability of the watermark.  Further even for such schemes, the estimates obtained by considering only the watermark probabilities are observed to be practically very useful in most cases, without the need to consider the updated probabilities $p_k$. We show in results that the error magnitude is quite small and we achieve descent approximations by only considering $p \approx p_W$. Thus unless extremely accurate estimates are required, we may continue to use the tree based algorithm in section~\ref{treealgo} to achieve quick estimates.

\subsection{ Estimating the optimal number of passes}\label{opt}
The methods we have discussed upto this point have concentrated on providing estimates of the embedding capacity for a given number of passes $P$. We can however find the overall embedding capacity and the optimal number of passes as well. For typical watermarking schemes the stage-wise embedding capacity keeps on decreasing. This is due to an increase in the number of pixel pairs not eligible for embedding at every successive pass of watermarking. Further at every stage the overhead in the form of the size of the auxiliary data also generally increases. Thus the optimal number of passes is reached when the stage-wise embedding capacity $\watbit_k$ touches zero. 

\subsection{Generalization to watermarking schemes operating on groups of pixels}
We have considered upto now watermarking schemes operating on pairs of pixels only, and hence we have assumed $\xi = (x,y)$ in our algorithms. However our methods can be extended to consider $n -$tuple of pixels $\xi = (x^1, x^2, \cdots, x^n)$. Consider the case of $n = 3$ or triplets of pixels~\cite{alattar:reversible}. Here in every triplet $2$ bits can be embedded, so we need to consider $4$ cases of embedding $00, 01, 10$ and $11$. Let the transformed pixel triplets be $\xi^{00}, \xi^{01}, \xi^{10}$ and $\xi^{11}$. If a pixel triplet is not embeddable it will transform to $\xi^{\phi}$. Then we can easily extend the co-occurrence based method, by considering here that $(1-p)^2, p(1-p), p(1-p)$ and $p^2$ fraction of the pixel triplets of $C_k(\xi)$ transform to $C_{k+1}(\xi^{00}), C_{k+1}(\xi^{01}),C_{k+1}(\xi^{10})$ and $C_{k+1}(\xi^{11})$ respectively. The tree based method can similarly be estimated to consider a tree of groups of pixels instead of pixel pairs. The tree will have more branches. For example in the case of pixel triplets, each embeddable pixel triplet will have 4 possible paths to $\xi^{00}, \xi^{01}, \xi^{10}$ and $\xi^{11}$ respectively. We can however define a recursive formulation similar to equation~\eqref{totalexrec} only replacing the equation for the embeddable case as: 
$\EP(\xi,\genbit) =\size({b}_{\xi}) + p^2 \EPp(\xi^{11},\genbit) +  p(1-p)(\EPp(\xi^{01},\genbit) + \EPp(\xi^{10},\genbit))+ (1-p)^2 \EPp(\xi^{00},\genbit)$. Our methods can similarly be extended to deal with any arbitrary $n$, where-in every embeddable pixel $n-$tuple can get converted to $2^{n-1}$ possible pixel $n-$tuples since within every such tuple $n-1$ bits can be embedded.

\section{Bounds on the Embedding Capacity}
In this section we provide the bounds on the embedding capacity. The maximum embedding capacity of a given image can be thought of as the largest possible watermark which can be embedded into that image. In other words no watermark with a size is larger than this, can be embedded completely into that image. We use the same tree based implementation discussed in section~\ref{treealgo} to find the maximum possible embedding capacity, by considering every possible type of watermark. We denote the largest size as $\size^{\max}(.)$. Then using the notation developed in section~\ref{treealgo}, we can write:
\begin{equation}\label{maxbasic}
\eta^{\max}(\watbit) = \max_{\mathbf{s} \in \mathbb{S}}   \biggl[ \size(\embbit, \mathbf{s}) -  \size(\flagbit, \mathbf{s}) -  \size(\complocbit, \mathbf{s}) -  \size(\compflagbit, \mathbf{s}) \biggl]
\end{equation} 
Recall that $\mathbf{s} \in \mathbb{S}$ here represents $s_1 \in \mathbf{S}_{\xi_1}, s_2 \in \mathbf{S}_{\xi_2}, \cdots, s_N \in \mathbf{S}_{\xi_N}$.  Note that though not shown, $\eta^{\max}(\watbit)$ is a function of $P$, and it represents the maximum possible embedding capacity obtained in $P$ passes. In order to find the maximum possible embedding capacity in the image, we need to find $\underset{P}{\max} \mbox{ } \eta^{\max}(\watbit)$. As discussed in section \ref{opt}, the total embedding capacity initially increases till the optimal number of passes after which it begins to decrease. Hence we can simply start from $P = 0$, till the value of $P$ where $\eta^{\max}(\watbit)$ begins to decrease. The corresponding value of $\eta^{\max}(\watbit)$ will be the maximum possible embedding capacity for that image.

\begin{theorem}
The maximum possible embedding capacity $\eta^{\max}(\watbit)$ upto $P$ passes can be upper bounded by:
 \begin{eqnarray}
\eta^{\max}(\watbit) \le \sum_{j = 1}^N (\mathbf{M}_P^{\size}(\xi_j, \embbit) - \mathbf{M}_P^{\size}(\xi_j, \flagbit)) - \sum_{k = 0}^{P-1} N H_0\biggl(\frac{\overset{N}{\underset{j = 1}{\sum}} \mathbf{M}^{\ones}(\xi_j, \locbit_k)}{N}\biggl) 
- \sum_{k = 0}^{P-1} \biggl( \sum_{j = 1} ^ N \mathbf{L}^{\size}(\xi_j, \mathcal{F}_k) \biggl)  \min(\alpha_k, \beta_k)
 \end{eqnarray}
with:
\begin{eqnarray}
\alpha_k = H_0\biggl(\frac{\overset{N}{\underset{j = 1}{\sum}} \mathbf{M}^{\ones}(\xi_j, \flagbit_k)}{ \overset{N}{\underset{j = 1}{\sum}} \mathbf{M}^{\ones}(\xi_j, \flagbit_k) + \mathbf{L}^{\ones}(\xi_j, \bar{\flagbit_k})}\biggl) , \,\,\,
\beta_k = H_0\biggl(\frac{\overset{N}{\underset{j = 1}{\sum}} \mathbf{L}^{\ones}(\xi_j, \flagbit_k)}{ \overset{N}{\underset{j = 1}{\sum}} \mathbf{L}^{\ones}(\xi_j, \flagbit_k) + \mathbf{M}^{\ones}(\xi_j, \bar{\flagbit_k})}\biggl)
\end{eqnarray}
and for a general bitstream $\genbit$: 
\begin{eqnarray}\label{maxqtys}
\mathbf{M}_P^{\size}(\xi, \genbit) = \underset{s \in \mathbf{S}_{\xi}}{\max} \mbox{ } \overset{P-1}{\underset{k = 0}{\sum}} \size(b_{s[k]}) , \,\,\,  \mathbf{M}^{\ones}(\xi, \mathcal{B}_k) =  \underset{s \in \mathbf{S}_{\xi}}{\max}\mbox{ } b_{s[k]} \\ 
\mathbf{L}^{\size}(\xi, \mathcal{B}_k) =  \underset{s \in \mathbf{S}_{\xi}}{\min}\mbox{ } \size(b_{s[k]}) , \,\,\,
\mathbf{L}^{\ones}(\xi, \mathcal{B}_k) =  \underset{s \in \mathbf{S}_{\xi}}{\min}\mbox{ } b_{s[k]}.  
\end{eqnarray}
Further the bitstream $\bar{\flagbit_k}$ is formed by inverting every bit in $\flagbit_k$. 
\end{theorem} 
\begin{proof}
We start with equation~\eqref{maxbasic} and slowly relax the constraints as follows:
\begin{eqnarray}
\eta^{\max}(\watbit) &=& \max_{\mathbf{s} \in \mathbb{S}}   \biggl[ \size(\embbit, \mathbf{s}) -  \size(\flagbit, \mathbf{s}) -  \size(\complocbit, \mathbf{s}) -  \size(\compflagbit, \mathbf{s}) \biggl] \nonumber\\
					&\le& \max_{\mathbf{s} \in \mathbb{S}}  \biggl[ \size(\embbit, \mathbf{s}) -  \size(\flagbit, \mathbf{s}) \biggl] - \min_{\mathbf{s} \in \mathbb{S}} \size(\complocbit, \mathbf{s}) - \min_{\mathbf{s} \in \mathbb{S}} \size(\compflagbit, \mathbf{s}). \nonumber
\end{eqnarray}
We now consider the first two terms in the above expression. We can expand it as follows:
\begin{eqnarray}\label{embmax}
&& \max_{\mathbf{s} \in \mathbb{S}} \biggl[ \size(\embbit, \mathbf{s}) -  \size(\flagbit, \mathbf{s}) \biggl] = \max_{\mathbf{s} \in \mathbb{S}} \sum_{j = 1} ^ N  \biggl[ \size(\xi_j, s_j, \embbit) -  \size(\xi_j, s_j,\flagbit) \biggl]  \nonumber \\
					&=& \sum_{j = 1} ^ N \max_{s_j \in \mathbf{S}_{\xi_j}} \biggl[ \size(\xi_j, s_j, \embbit) -  \size(\xi_j, s_j,\flagbit) \biggl] = \sum_{j = 1} ^ N \max_{s_j \in \mathbf{S}_{\xi_j}} \biggl[ \sum_{k = 0}^{P-1} \size(i_{s[k]}) - \sum_{k = 0}^{P-1} \size(f_{s[k]}) \biggl] \nonumber \\
					&=&  \sum_{j = 1} ^ N  (\mathbf{M}_P^{\size}(\xi_j, \embbit) - \mathbf{M}_P^{\size}(\xi_j, \flagbit)).
\end{eqnarray}
Now we consider the term involving the compressed location map. Since this term appears with a negative sign, finding an upper bound on $\eta^{\max}(\watbit)$ is equivalent to finding a lower bound of this term. Hence we first replace the compression term by the entropy of the location map bitstream at every stage, which is a lower bound to the size of the compressed bitstream. Hence we can write:
\begin{eqnarray}
\min_{\mathbf{s} \in \mathbb{S}} \size(\complocbit, \mathbf{s}) &\ge& \min_{\mathbf{s} \in \mathbb{S}} \sum_{k = 0}^{P-1} N H_0(\frac{\ones(\mathcal{L}_k, \mathbf{s})}{N}) \ge \min_{\mathbf{s} \in \mathbb{S}} \sum_{k = 0}^{P-1}  N H_0(\frac{\sum_{j = 1}^ N l_{s_j[k]}}{N}) \nonumber \\
											&\ge& \sum_{k = 0}^{P-1} \min_{\mathbf{s} \in \mathbb{S}} N H_0(\frac{\sum_{j = 1}^ N l_{s_j[k]}}{N}). \nonumber 
\end{eqnarray}
We then further simplify the above expression, using the lemma \ref{lemma1} (Provided in Appendix-B). Observe that $\mathcal{L}_k$ is typically dominated by ones initially since most of the pixel pairs are embeddable. Further compression is only possible till $\ones(\mathcal{L}_k, \mathbf{s}) \ge N/2$ since, as the fraction of the number of ones reaches 0.5, the compression factor reaches 1 and hence embedding is no longer possible. Thus $\frac{\ones(\mathcal{L}_k, \mathbf{s})}{N} > 0.5, \forall k$ and hence using lemma~\ref{lemma1} directly with $f(\mathbf{s}) = \frac{\ones(\mathcal{L}_k, \mathbf{s})}{N}$ we have:
\begin{eqnarray}\label{locmax}
\min_{\mathbf{s} \in \mathbb{S}} \size(\complocbit, \mathbf{s}) &\ge& \sum_{k = 0}^{P-1} N H_0(\frac{\max_{\mathbf{s} \in \mathbb{S}} \sum_{j = 1}^ N l_{s_j[k]}}{N}) \ge \sum_{k = 0}^{P-1} N H_0(\frac{\sum_{j = 1}^ N \max_{s_j \in \mathbb{S}_{\xi_j}} l_{s_j[k]}}{N}) \nonumber \\
							&\ge& \sum_{k = 0}^{P-1} N H_0(\frac{\sum_{j = 1}^ N \mathbf{M}^{\ones}(\xi_j, \mathcal{L}_k)}{N}) 
\end{eqnarray}
We now finally consider the term involving the compressed flag bitstream. Again similar to the compressed location map size estimation, we can write:
\begin{eqnarray}
\min_{\mathbf{s} \in \mathbb{S}} \size(\compflagbit, \mathbf{s}) &\ge& \min_{\mathbf{s} \in \mathbb{S}} \sum_{k = 0}^{P-1} \size(\mathcal{F}_k,\mathbf{s}) H_0(\frac{\ones(\mathcal{F}_k, \mathbf{s})}{\size(\mathcal{F}_k,\mathbf{s})}) \nonumber \\
											&\ge&\sum_{k = 0}^{P-1}  \min_{\mathbf{s} \in \mathbb{S}} \size(\mathcal{F}_k,\mathbf{s}) H_0(\frac{\ones(\mathcal{F}_k, \mathbf{s})}{\size(\mathcal{F}_k,\mathbf{s})}) \nonumber \\
											&\ge&\sum_{k = 0}^{P-1}  \min_{\mathbf{s} \in \mathbb{S}} \size(\mathcal{F}_k,\mathbf{s})  \min_{\mathbf{s} \in \mathbb{S}} H_0(\frac{\ones(\mathcal{F}_k, \mathbf{s})}{\size(\mathcal{F}_k,\mathbf{s})}) \nonumber
\end{eqnarray}
We now consider each of the terms above separately. Consider the first term in the above expression.
\begin{eqnarray}\label{flagmax1}
\min_{\mathbf{s} \in \mathbb{S}} \size(\mathcal{F}_k,\mathbf{s}) &=& \min_{\mathbf{s} \in \mathbb{S}} \sum_{j = 1} ^ N \size(f_{s_j[k]}) \nonumber \\
									&=&  \sum_{j = 1} ^ N \min_{s_j \in \mathbb{S}_{\xi_j}} \size(f_{s_j[k]}) =  \sum_{j = 1} ^ N \mathbf{L}^{\size}(\xi_j, \mathcal{F}_k) 
\end{eqnarray}
Now we consider the second term. Again we take the min term inside the entropy. However unlike the case of the compressed location map, we cannot comment on the value of $\frac{\ones(\mathcal{F}_k, \mathbf{s})}{\size(\mathcal{F}_k,\mathbf{s})}$, whether it lies within $[0.5,1]$ or not. Hence we cannot use lemma-\ref{lemma1} directly. However we use lemma~\ref{lemma2} (Again provided in Appendix B), and define $h(\mathbf{s}) = \frac{\ones(\mathcal{F}_k, \mathbf{s})}{\size(\mathcal{F}_k,\mathbf{s})}$ to obtain:
\begin{eqnarray}\label{flagmax2}
\min_{\mathbf{s} \in \mathbb{S}} H_0(\frac{\ones(\mathcal{F}_k, \mathbf{s})}{\size(\mathcal{F}_k,\mathbf{s})}) \ge \min \biggl( H_0(\max_{\mathbf{s} \in \mathbb{S}} \frac{\ones(\mathcal{F}_k, \mathbf{s})}{\size(\mathcal{F}_k,\mathbf{s})} ),
	  H_0(\min_{\mathbf{s} \in \mathbb{S}} \frac{\ones(\mathcal{F}_k, \mathbf{s})}{\size(\mathcal{F}_k,\mathbf{s})} \biggl)
\end{eqnarray}
Consider the first term within the entropy. We use a trick here of defining $\bar{f_{s[k]}} = 1 - f_{s[k]}$, if $s[k] \in \flagdom$ and $0$ otherwise. Then observe that $\size(f_{s_j[k]}) = f_{s_j[k]} + \bar{f}_{s_j[k]}$. Thus we can write:

\begin{eqnarray}\label{flagmax3}
H_0 \biggl(\max_{\mathbf{s} \in \mathbb{S}} \frac{\ones(\mathcal{F}_k, \mathbf{s})}{\size(\mathcal{F}_k,\mathbf{s})} \biggl) &=&  H_0\biggl( \max_{\mathbf{s} \in \mathbb{S}} \frac{\sum_{j = 1} ^ N f_{s_j[k]}}{\sum_{j = 1} ^ N \size(f_{s_j[k]})}\biggl) \nonumber =  H_0\biggl( \max_{\mathbf{s} \in \mathbb{S}} \frac{1}{1 + \frac{\sum_{j = 1} ^ N\bar{f}_{s_j[k]}}{\sum_{j = 1} ^ N f_{s_j[k]}}} \biggl) \nonumber \\
			&=&  H_0\biggl( \frac{1}{1 +  \min_{\mathbf{s} \in \mathbb{S}} \frac{\sum_{j = 1} ^ N\bar{f}_{s_j[k]}}{\sum_{j = 1} ^ N f_{s_j[k]}}} \biggl) \ge H_0\biggl( \frac{1}{1 +  \frac{ \min_{\mathbf{s} \in \mathbb{S}} \sum_{j = 1} ^ N\bar{f}_{s_j[k]}}{ \max_{\mathbf{s} \in \mathbb{S}} \sum_{j = 1} ^ N f_{s_j[k]}}} \biggl) \nonumber \\
			&\ge& H_0\biggl( \frac{1}{1 +  \frac{ \sum_{j = 1} ^ N \min_{s_j \in \mathbb{S}_{\xi_j}} \bar{f}_{s_j[k]}}{ \sum_{j = 1} ^ N \max_{s_j \in \mathbb{S}_{\xi_j}} f_{s_j[k]}}} \biggl) \ge H_0\biggl( \frac{1}{1 +  \frac{ \sum_{j = 1} ^ N \mathbf{L}^{\ones}(\xi_j, \bar{\mathcal{F}}_k)}{ \sum_{j = 1} ^ N \mathbf{M}^{\ones}(\xi_j, \mathcal{F}_k)}} \biggl) 
\end{eqnarray}
Similarly consider the second term and using derivations similar to above, we can obtain:
\begin{equation}
H_0 \biggl(\min_{\mathbf{s} \in \mathbb{S}} \frac{\ones(\mathcal{F}_k, \mathbf{s})}{\size(\mathcal{F}_k,\mathbf{s})} \biggl) \ge H_0\biggl( \frac{1}{1 +  \frac{ \sum_{j = 1} ^ N \mathbf{M}^{\ones}(\xi_j, \bar{\mathcal{F}}_k)}{ \sum_{j = 1} ^ N \mathbf{L}^{\ones}(\xi_j, \mathcal{F}_k)}} \biggl) 
\end{equation}
Thus the theorem is proved by combining all the above equations.
\end{proof}

Estimation of the maximum embedding capacity, requires estimation of the quantities in equation~\eqref{maxqtys}. They however require exhaustive search over the pixel pair tree for every pixel pair, due to which the computation becomes extremely expensive. We can however recursively reformulate these expressions, similar to equations~\eqref{stageexrec} and \eqref{totalexrec}. Thus we can recursively compute $\mathbf{M}_P^{\size}(\xi, \genbit)$ for $P > 1$ as:
\begin{equation} \label{abcde}
\mathbf{M}_P^{\size}(\xi, \genbit) =
\left\{
	\begin{array}{llll}
		\size(b_{\xi}) + \max(\mathbf{M}_{P-1}^{\size}(\xi^0, \genbit), \mathbf{M}_{P-1}^{\size}(\xi^1, \genbit)), & \mbox{if }\xi \in \embdom\\
		\size(b_{\xi}) + \mathbf{M}_{P-1}^{\size}(\xi^{\phi}, \genbit), & \mbox{if } \xi \notin \embdom 
	\end{array}
\right.
\end{equation}
The base case here is $\mathbf{M}_1^{\size}(\xi, \genbit) =\size(b_{\xi})$. Similarly we reformulate the expression for $\mathbf{M}^{\ones}(\xi, \mathcal{B}_k)$ for $k > 0$ as
\begin{equation} \label{abcd}
\mathbf{M}^{\ones}(\xi, \mathcal{B}_k) =
\left\{
	\begin{array}{llll}
		 \max(\mathbf{M}^{\ones}(\xi^0, \mathcal{B}_{k-1}),  \mathbf{M}^{\ones}(\xi^1, \mathcal{B}_{k-1})), & \mbox{if }\xi \in \embdom\\
		\mathbf{M}^{\ones}(\xi^{\phi}, \mathcal{B}_{k-1}), & \mbox{if } \xi \notin \embdom 
	\end{array}
\right.
\end{equation}
Again the base case here is $\mathbf{M}^{\ones}(\xi, \mathcal{B}_0) = b_{\xi}$.

Note that the expression for $\mathbf{L}^{\ones}(\xi, \mathcal{B}_k)$ is similar to equation~\eqref{abcd} only replacing the $\max(.)$ by $\min(.)$. The above expressions can easily be proved using methods similar to those used in theorem-1. Also we can use tricks like memoization etc. similar to those discussed in section-\ref{treealgo}, to efficiently estimate the maximum embedding capacity. 

\section{Results and Discussion}
\begin{figure}
\centering
\subfloat[]{
\setlength\fboxsep{0pt}
\fbox{\includegraphics[height=3.5cm,width=3.5cm]{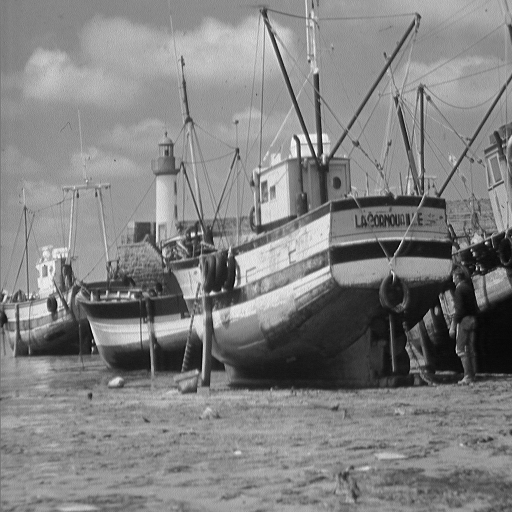}}
\label{coltuc1}
}
\subfloat[]{
\setlength\fboxsep{0.5pt}
\fbox{\includegraphics[height=3.5cm,width=3.5cm]{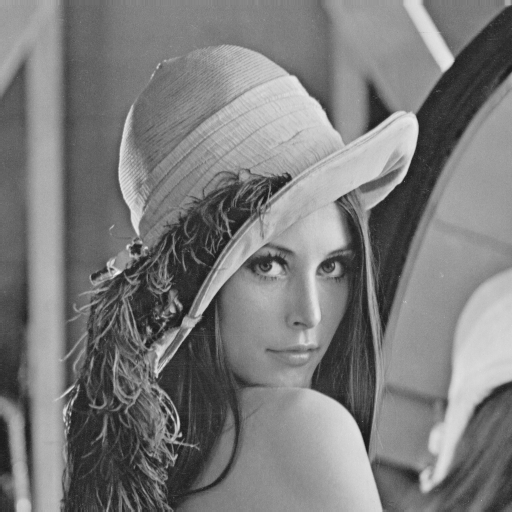}}
\label{coltuc2}
}
\caption{Above are the Boat (left) and Lena (right) images on which we have performed our analysis.} 
\label{analimages}
\end{figure} 
The presented concept has been experimented on several commonly known images as possible cover images. Images were chosen with varying gray level 
distributions. We however present results here only for Lena and Boat images  for brevity (shown in fig.~\ref{analimages}). A general analysis for several other images has been provided as additional material\footnotemark.  \footnotetext[1]{Additional results are provided in \url{''http://www.ee.iitb.ac.in/~sc/main/misc/results.html''}}
We compute the embedding capacity for these cover images for watermarks of varying number of ones and zeros and compare them to the actual capacity as obtained through direct embedding of these watermarks in the corresponding cover image. We also find the bounds on the embedding capacity for these images. Lastly the codes used in the quantitative evaluation of our results shall be made available once our paper has been accepted.
\subsection{Dependance of Embedding capacity on probability of watermark alone}\label{exver}
In this section, we statistically verify our claim, that the multipass embedding capacity of pixel-pair based watermarking schemes depends only on the probability mass function of the embedding bitstream and not on the actual bitstream itself. In particular, we verify this for the Lena image, for the schemes of Coltuc~\cite{coltuc:RCM} and Tian~\cite{tian:reversible}, and observe that the variance of the embedding capacity obtained by taking 100 different watermarks with $p=0.6$ is negligible. The values of the standard deviations of the bpp's at various stages of watermarking are shown in Table I.
\begin{table}
\centering
\caption{The standard deviation of the total embedding capacity at the end of various watermarking stages (in $10^{-2}$ bpp).}
\label{table1}
\begin{tabular}{|c|c|c|c|l|} \hline
\textbf{Algo}&\textbf{Stage-1}&\textbf{Stage-2}&\textbf{Stage-3}\\ 
\hline \ {Coltuc}&0.3&1.1&1.6\\ \hline
\ {Tian}&0.3&1.3&1.5\\ 
\hline\end{tabular}
\label{timinganalysis}
\end{table}
\begin{figure}[t]
\centering
\subfloat{
\includegraphics[height=5cm]{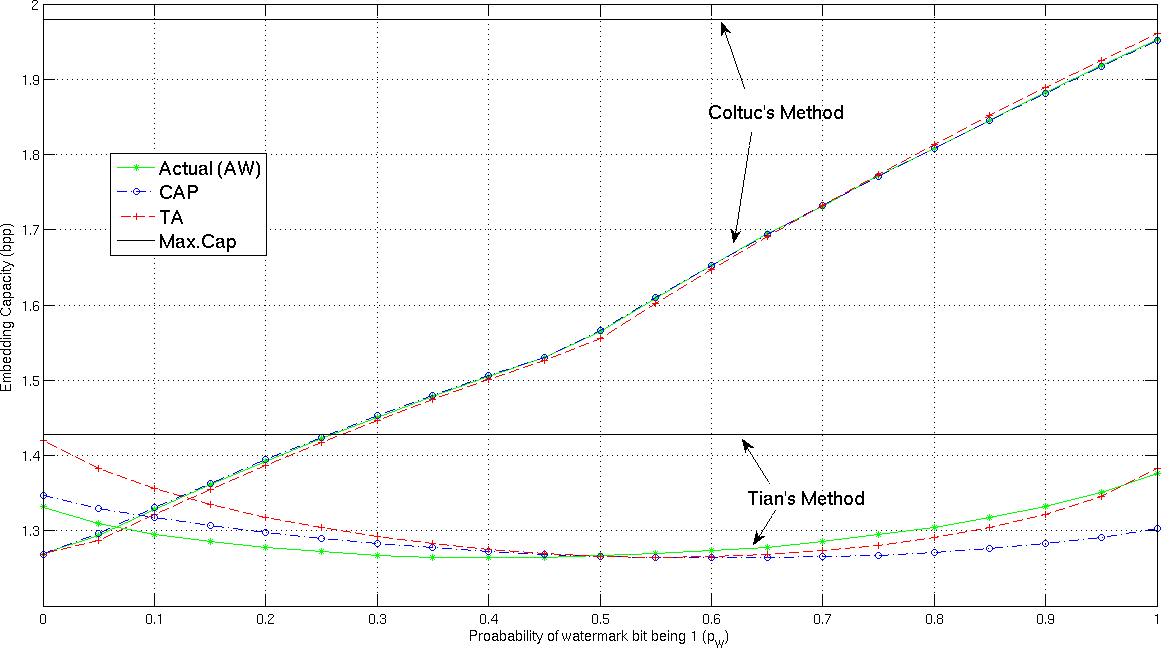}
}
\caption{ The  capacity estimation for Lena image obtained through Tian's and Coltuc's schemes. For each of these watermarking schemes, the results of AW, CAP, TA and Maxcap are shown.}
\label{lenaresults}
\end{figure}	
\subsection{Estimation and bounds of embedding capacity}
\begin{figure}[t]
\subfloat{
\includegraphics[ width=11cm, height = 6cm]{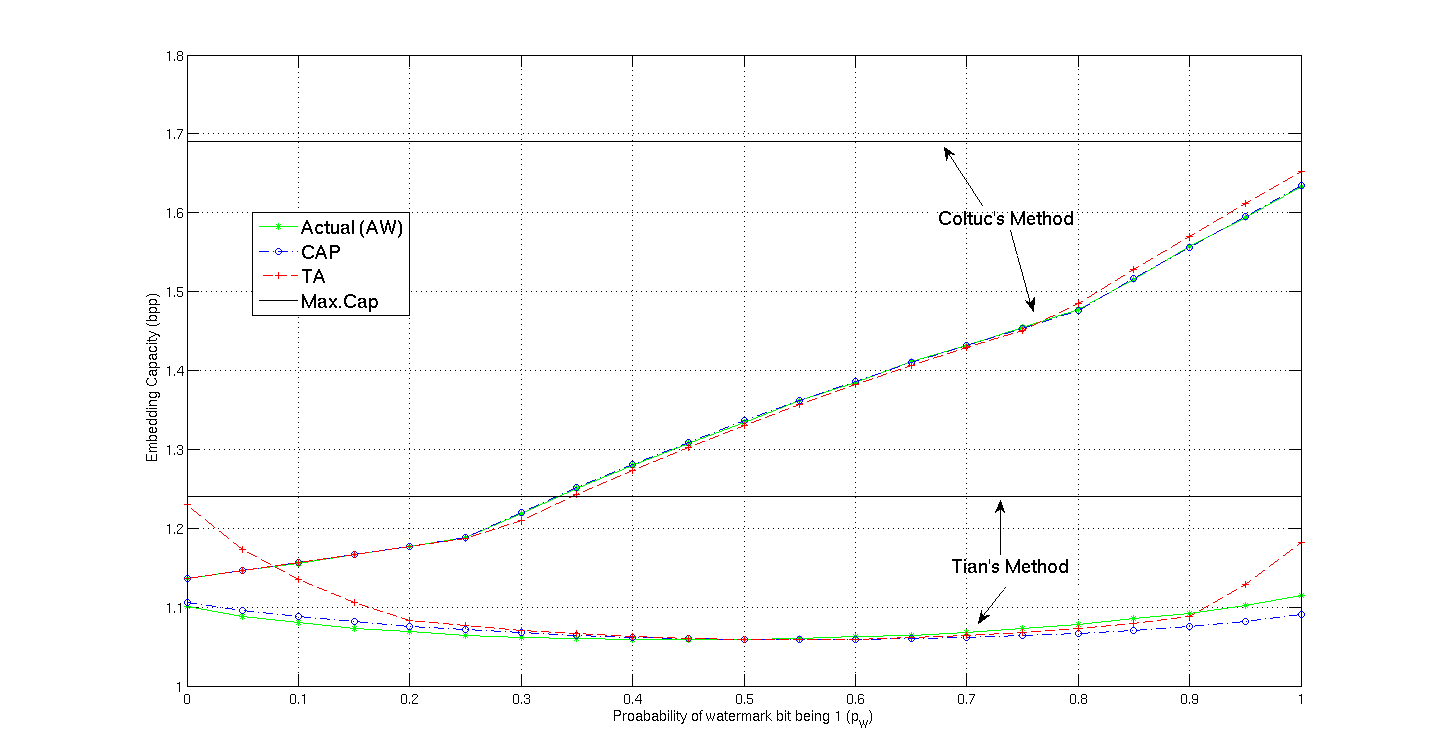}
}
\caption{ The results for Boat image obtained through Tian's and Coltuc's schemes.}
\label{hillresults}
\end{figure}
We demonstrate in this section, how the embedding capacity estimates obtained using our algorithms, are reasonably close to those obtained by actual embedding a watermark, for the schemes of Tian~\cite{tian:reversible} and Coltuc~\cite{coltuc:RCM}. In particular we compare the capacities obtained by actually watermarking (AW) with the estimates obtained by the tree based algorithm (TA), along with the Co-occurrence matrix based adaptive probability algorithm (CAP) discussed in section \ref{cap}. Finally we also compute the maximum possible embedding capacity (MaxCap), discussed in section IV. The CAP algorithm is precise since it considers weighted probabilities discussed in section \ref{cap}, while TA, though not as accurate is extremely fast . We have computed here the optimal embedding capacity, as discussed in section \ref{opt}. Lastly observe that we compute these for various watermarks, characterized by the fraction of the number of ones in the watermark ($p_W$), and are plotted in fig.~\ref{lenaresults} and fig.~\ref{hillresults}. 

\subsubsection{Coltuc et al.} Coltuc's algorithm does not require any location map, and requires flag bits to be embedded along with the watermark. Hence the estimates provided are very precise and accurate. In particular, the estimates provided by the CAP algorithm are extremely accurate and consistent with the actual capacity for all probabilities. Further the estimates provided by the tree based algorithm (TA) assuming $p = p_W$, are reasonably close as evident in  fig.~\ref{lenaresults} and fig. \ref{hillresults}.

\subsubsection{Tian} Tian's algorithm requires a compression stage, and correspondingly we estimate the compression, using entropy based method discussed in section~\ref{compstream}. In order to achieve higher embedding, typically the flag bits are also compressed. Further we assume that the bitstreams are compressed using arithmetic coding. As evident from fig.~\ref{lenaresults} and fig. \ref{hillresults}, the estimates provided for Tian's schemes are not as precise as the ones provided by coltuc, mainly due to our approximation of the compression of the flag and location map bit streams. Nevertheless, except for the extreme probabilities, the estimates for most watermark bit probabilities are quite precise and useful. Notwithstanding this, it is also evident that CAP estimates are better than those of the TA algorithm.

\subsection{Comparative Analysis} The proposed algorithms not only provide useful estimates of embedding capacity but are also extremely fast. The CAP algorithm is a more general algorithm since it is easily amenable to updating probabilities at every stage, and takes account of the contribution of the auxiliary bits. The tree  based algorithm is however not amenable to updating probabilities at every iteration, but is extremely fast. Further note that the timings given in Table-\ref{table:5} for actual watermarking (AW), are the ones obtained by the most efficient look-up-table based implementation of Coltuc and Tian's algorithm. It is evident that our algorithms are computationally much more efficient than these low cost implementations. In particular, the tree based algorithm (TA) is extremely fast, nearly 200 to 400 times faster than actual watermarking, while the CAP algorithm, though iterative is also reasonably fast, and is about 100 times faster than the actual watermarking. The main reason for this is the computational efficiency gained through the sparse representation of the co-occurrence matrix and the simple computations of our algorithms, vis-a-vis the complexity of actual watermarking due to the numerous iterations of collecting, embeding and possibly compressing the auxiliary data along with the watermark. Also as evident from the timings, Tian's algorithm is computationally expensive since it involves a data compression stage. 
\begin{table}
\centering
\caption{Timing analysis for different algorithms~(time taken in seconds) for execution in MATLAB for a 3.2~GHz PC with 2~GB RAM.}
\label{table:5}
\begin{tabular}{|c|c|c|c|l|} \hline
\textbf{Algo}&\textbf{Image}&\textbf{AW}&\textbf{CAP}&\textbf{TA}\\ 
\textbf{}&\textbf{Size}&{}&{}&{}\\ 
\hline \ {}&256 x 256&11.21&0.17&0.06\\ 
\ {Coltuc}&512 x 512&31.32&0.35&0.19\\ 
\ {}&1024 x 1024&96.65&0.43&0.25\\\hline
 \ {}&256 x 256&37.23&0.21&0.14\\ 
\ {Tian}&512 x 512&99.32&0.55&0.31\\ 
\ {}&1024 x 1024&172.12&0.73&0.38\\ 
\hline\end{tabular}
\label{timinganalysis}
\end{table}
\section{Conclusion} This is mainly a theoretical paper, centered around multipass embedding capacity estimation for mapping and expansion based algorithms. We provide very general purpose algorithms, which can be applied to any watermarking scheme, operating on blocks of pixels. In particular, we implemented a co-occurrence based and a tree based algorithm, to efficiently estimate the capacity, at computational costs substantially lesser than actually embedding the watermark. We further show in Appendix-A, that both of these algorithms provide the same estimates. From our results, it is evident that the tree based method is computationally more efficient compared to co-occurrence based one since it requires only a single iteration. The co-occurrence based method on the other hand is however a more general framework and can be extended to consider different watermark probabilities at every iteration and correspondingly can be used to provide accurate estimates by considering the probabilities of the auxiliary bits as well.

\section{Acknowledgement} We would like to thank Subhasis Das, Siddhant Agrawal and Prof. Sibi-Raj Pillai from IIT Bombay for suggestions and discussions. Finally we gratefully acknowledge partial funding from Bharati Centre for Communication.

\bibliographystyle{abbrvnat}
\bibliography{reversible} 

\begin{thebibliography}{10}

\bibitem{alattar:reversible}
A.M.Alattar.
\newblock Reversible watermark using the difference expansion of a generalized
  integer transform.
\newblock {\em IEEE Trans. Image Processing}, 13(8):1147--1156, August 2004.

\bibitem{borse:capacity}
R.~Borse and S.~Chaudhuri.
\newblock Computation of embedding capacity in reversible watermarking schemes.
\newblock {\em Proceedings of the ACM's ICVGIP-10, Chennai, India}, 2010.

\bibitem{cover1991elements}
T.~Cover, J.~Thomas, and MyiLibrary.
\newblock {\em Elements of information theory}, volume~6.
\newblock Wiley Online Library, 1991.

\bibitem{ingemar:digital}
I.~Cox.
\newblock {\em Digital Watermarking and Steganography}.
\newblock Morgan Kaufmann,Second Edition, New York, 2008.

\bibitem{coltuc:highcapacity}
D.Coltuc and J.-M. Chassery.
\newblock High capacity reversible watermarking.
\newblock {\em Proceedings of the IEEE International Conference on Image
  Procesing ICIP}, pages 2565--2568, 2006.

\bibitem{coltuc:RCM}
D.Coltuc and J.-M. Chassery.
\newblock Very fast watermarking by reversible contrast mapping.
\newblock {\em IEEE Signal processing Letters}, 14(4):255--258, 2007.

\bibitem{thodi:expansion}
D.M.Thodi and J.J.Rodriquez.
\newblock Expansion embedding techniques for reversible watermarking.
\newblock {\em IEEE Trans. Image Processing}, 16(3):721--730, 2007.

\bibitem{feng:survey}
J.~Feng, I.~Lin, C.~Tsai, and Y.~Chu.
\newblock Reversible watermarking:current status and key issues.
\newblock {\em International Journal of Network Security}, 12(3):161--171,
  2006.

\bibitem{golub1996matrix}
G.~Golub and C.~Van~Loan.
\newblock Matrix computations 3rd edition.
\newblock {\em The John Hopkins University, Baltimore}, 1996.

\bibitem{haralick:cooccur}
R.~Haralick, K.~Shanmugam, and I.~Dinstein.
\newblock {Textural features for image classification}.
\newblock {\em IEEE Transactions on systems, man and cybernetics},
  3(6):610--621, 1973.

\bibitem{tian:reversible}
J.Tian.
\newblock Reversible data embedding using a difference expansion.
\newblock {\em IEEE Trans. Circuits Systems for Video Technology},
  13(8):890--896, August 2003.

\bibitem{kalker:bounds}
T.~Kalker and F.~M.J.Willems.
\newblock Capacity bounds and constructions for reversible data-hiding.
\newblock {\em IEEE International Conference on DSP}, 2002.

\bibitem{li:reversible}
C.-T. Li.
\newblock Reversible watermarking scheme with image-independent embedding
  capacity.
\newblock {\em IEE Trans on Vision and Image Processing}, 152(6)(779-786),
  2005.

\bibitem{kamastra:reversible}
L.Kamastra and H.Heijmans.
\newblock Reversible data embedding into images using wavelet techniques and
  sorting.
\newblock {\em IEEE Trans. Image Processing}, 14(12):2082--2090, December 2005.

\bibitem{michie1968memo}
D.~Michie.
\newblock Memo functions and machine learning.
\newblock {\em Nature}, 218(1):19--22, 1968.

\bibitem{celik:lossless}
M.U.Celik and G.Sharma.
\newblock Lossless generalized-lsb data embedding.
\newblock {\em IEEE Transactions on Image Processing}, 14(2):253--266, 2005.

\bibitem{pissanetzky1984sparse}
S.~Pissanetzky.
\newblock {\em Sparse matrix technology}.
\newblock Academic Press London, 1984.

\bibitem{caldelli2010reversible}
C.~Roberto, F.~Francesco, and B.~Rudy.
\newblock Reversible watermarking techniques: An overview and a classification.
\newblock {\em EURASIP Journal on Information Security}, 2010.

\bibitem{sachnev:reversible}
V.~Sachnev.
\newblock Reversible watermarking algorithm using sorting and prediction.
\newblock {\em IEEE Trans. Circuits Systems for Video Technology},
  19(7):989--999, July 2009.

\bibitem{shannon271948}
C.~Shannon.
\newblock A mathematical theory of communication.
\newblock {\em Bell System Technical Journal}, 27(10):623--656.

\bibitem{binshift}
C.~D. Vleeschouwer, J.~E. Delaigle, and B.~Macq.
\newblock Circular interpretation of histogram for reversible watermarking.
\newblock {\em in Proceedings of the IEEE 4th Workshop on Multimedia Signal
  Processing}, pages 345--350, Oct 2001.

\bibitem{weng2008reversible}
S.~Weng, Y.~Zhao, J.~Pan, and R.~Ni.
\newblock Reversible watermarking based on invariability and adjustment on
  pixel pairs.
\newblock {\em Signal Processing Letters, IEEE}, 15:721--724, 2008.

\end{thebibliography}
\appendix
\section{Appendix}
\begin{theorem}\label{thm3}
Under multi-pass embedding, the embedding capacity estimated by the co-occurrence based algorithm (algorithm-1) is exactly the same as the one estimated by the pixel-pair tree based algorithm. In particular we prove that:
 \begin{equation*}\mbox{1) } \sum_{\xi \in \dom} C_k(\xi) b_{\xi} = \sum_{\xi \in \dom} C_0(\xi) \sum_{s \in \mathbf{S}_{\xi}} \biggl(~\prod_{m=0}^{k} {p_{s[m]}} \biggl)  {{b}_{s[k]}} \end{equation*}
\begin{equation*}\mbox{2) }  \sum_{k = 0}^{P-1} \sum_{\xi \in \dom} C_k(\xi) b_{\xi} = \sum_{\xi \in \dom} C_0(\xi) \sum_{s \in \mathbf{S}_{\xi}} \biggl(~\prod_{k=0}^{P-1} {p_{s[k]}} \biggl) \sum_{k=0}^{P-1} {{b}_{s[k]}} \end{equation*}
\end{theorem} 
\begin{proof}
For every pixel pair $\xi$, let $\mathbf{N}_{\xi}$ represent the set of nodes corresponding to the pixel pair tree of $\xi$. Let $n \in \mathbf{N}_{\xi}$ refer to a node, i.e a pixel pair in the pixel-pair tree of $\xi$ . Further let $s_n$ represent any path of the tree, containing the node $n$ and let $d_n$ represent the depth of node $n$. Hence $s_n[d_n] = n$.  Now we prove the first part of the theorem. We start with the LHS and from algorithm~\ref{algo1}, it is evident that for any node $n \in \mathbf{N}_{\xi}$ with depth $d_n = k$, the number of pixel pairs contributed by the pair $\xi$, to the pixel pair $n$ is $C(\xi)\prod_{m = 0}^{k} p_{s_n[m]}$. This is clear from the updation step in algorithm~\ref{algo1}. We then break down the sum $\sum_{\xi \in \dom} C_k(\xi) b_{\xi}$, into contributions from every node $n \in \mathbf{N}_{\xi}$ with depth $d_n = k$, from every pixel pair in $\xi \in \dom$. In other words we can write:
\begin{equation} \label{fact1}
\sum_{\xi \in \dom} C_k(\xi) b_{\xi} = \sum_{\xi \in \dom} \sum_{n \in \mathbf{N}_{\xi} : d_n = k} C_0(\xi)\prod_{m = 0}^{k} p_{s_n[m]} b_{s_n[k]}
\end{equation}
Now it is obvious from the definition that: $ \sum_{n \in \mathbf{N}_{\xi} : d_n = k} \prod_{m = 0}^{k} p_{s_n[m]} b_{s_n[k]} = \sum_{s \in \mathbf{S}_{\xi}} \prod_{m=0}^{k} {p_{s[m]}} {b}_{s[k]}$. Hence using this in equation~\eqref{fact1} we can prove the first part of the theorem.

We now consider the second part. Using the result from earlier section we can write: 
\begin{eqnarray*}
\sum_{k = 0} ^ {P-1} \sum_{\xi \in \dom} C_k(\xi) b_{\xi} &=&  \sum_{\xi \in \dom} \sum_{k = 0} ^ {P-1} \sum_{n \in \mathbf{N}_{\xi} : d_n = k} C_0(\xi)\prod_{m = 0}^{k} p_{s_n[m]} b_{s_n[m]}\\
	&=& \sum_{\xi \in \dom} \sum_{n \in \mathbf{N}_{\xi}}  C_0(\xi)\prod_{k = 0}^{d_n} p_{s_n[k]} b_{s_n[d_n]}
\end{eqnarray*}
In the above expression, observe that we are summing over every node in the pixel pair tree as $k$ goes from $0$ to $P-1$. Hence we re-write the expression, by explicitly summing over every node in the pixel pair tree. We finally replace the variable $m$ by $k$, to obtain the above expression. In order to prove the second part of the theorem, it is sufficient to show that for every pixel pair $\xi$:
\begin{equation*}
 \sum_{n \in \mathbf{N}_{\xi}} C_0(\xi).\prod_{k = 0}^{d_n}p_{s_n[k]} b_n = \sum_{s \in \mathbf{S}_{\xi}} C_0(\xi) \prod_{k = 0}^{P-1} p_{s[k]} \sum_{k = 0}^{P-1} b_{s[k]}	
\end{equation*}
Note that we used the fact that $ b_{s_n[d_n]} = b_n$. We then start with the R.H.S as:
\begin{eqnarray*}
\mbox{R.H.S} &=& \sum_{s \in \mathbf{S}_{\xi}} C_0(\xi) \prod_{k = 0}^{P-1} p_{s[k]} \sum_{k = 0}^{P-1} b_{s[k]} \nonumber \\
		&=& \sum_{n \in \mathbf{N}_{\xi}} \sum_{s \in \mathbf{S}_{\xi}} C_0(\xi) \prod_{k = 0}^{P-1} p_{s[k]} \sum_{k = 0}^{P-1} b_{s[k]} I(s[k] = n) \nonumber \\
		&=& \sum_{n \in \mathbf{N}_{\xi}} \sum_{s \in \mathbf{S}_{\xi}} C_0(\xi) \prod_{k = 0}^{P-1} p_{s[k]} b_{n} I(s[k] = n) \nonumber 
\end{eqnarray*}
In the above equation, we consider only those paths which pass through the node $n$. Let $\mathbf{S}_n$ be a subset of paths which pass through node $n$. Hence we can write:
\begin{eqnarray*}
\mbox{R.H.S} &=& \sum_{n \in \mathbf{N}_{\xi}} \sum_{s \in \mathbf{S}_n} C_0(\xi) \prod_{k = 0}^{d_n} p_{s[k]}  \prod_{k = d_n + 1}^{P} p_{s[k]} b_n \nonumber \\
		&=& \sum_{n \in \mathbf{N}_{\xi}}  C_0(\xi) \prod_{k = 0}^{d_n} p_{s_n[k]} b_n \sum_{s \in \mathbf{S}_n} \prod_{k = d_n + 1}^{P} p_{s[k]} \nonumber \\
		&=&  \sum_{n \in \mathbf{N}_{\xi}}  C_0(\xi) \prod_{k = 0}^{d_n} p_{s_n[k]} b_n \nonumber \\
		&=& \mbox{L.H.S}
\end{eqnarray*}
Observe that the nodes in the paths $s \in \mathbf{S}_n$ with depth greater than $d_n$, form a complete subtree and hence we have:
$\sum_{s \in \mathbf{S}_n} \prod_{k = d_n + 1}^{P-1} p_{s[k]} = 1$. Further also observe that the nodes with a depth $\le d_n$ occur in every path in $s \in \mathbf{S}_n$. Hence we replace these nodes by $s_n[k]$, since $s_n$ is a path in $\mathbf{S}_n$.
Hence proved. Note that in this proof though we have proved the equivalence of the tree based and co-occurrence based algorithms for the estimates of $\ones(\genbit)$ and $\ones(\genbit_k)$, the same proof can be used to prove the equivalence of the estimates of $\size(\genbit)$ and $\size(\genbit_k)$.
\end{proof}

\begin{lemma}\label{lemma1}
For any function $f: \mathbb{S} \to [0.5,1]$, $\underset{\mathbf{s} \in \mathbb{S}}{\min} \mbox{ } H_0(f(\mathbf{s})) \equiv H_0(\underset{\mathbf{s} \in \mathbb{S}}{\max} f(\mathbf{s}))$. Similarly for any function $g: \mathbb{S} \to [0,0.5]$, $\underset{\mathbf{s} \in \mathbb{S}}{\min} \mbox{ } H_0(g(\mathbf{s})) \equiv H_0(\underset{\mathbf{s} \in \mathbb{S}}{\min}\mbox{ } g(\mathbf{s}))$.
\end{lemma}
\begin{proof}
First note that $H_0(z)$ is a decreasing function for $0.5 \le z \le 1$. Hence $H_0(z_1) \le H_0(z_2)$ implies $z_1 \ge z_2$. Using this fact, we now prove this lemma by contradiction. Assume that $\mathbf{s}_1 = \underset{\mathbf{s} \in \mathbb{S}}{\mbox{ arg min }} \mbox{ } H_0(f(\mathbf{s}))$ and $\mathbf{s}_2 = \underset{\mathbf{s} \in \mathbb{S}}{\mbox{ arg max }} \mbox{ } f(\mathbf{s})$. Further assume that $f(\mathbf{s}_1) \ne f(\mathbf{s}_2)$. Then clearly by definition $f(\mathbf{s}_1) \le f(\mathbf{s}_2)$ and $H_0(f(\mathbf{s}_1)) \le H_0(f(\mathbf{s}_2))$. This is a contradiction since $H_0(z)$ is a decreasing function. Thus $f(\mathbf{s}_1) \equiv f(\mathbf{s}_2) \Rightarrow \underset{\mathbf{s} \in \mathbb{S}}{\min} \mbox{ } H_0(f(\mathbf{s})) \equiv H_0(\underset{\mathbf{s} \in \mathbb{S}}{\max} f(\mathbf{s}))$. Similarly we can prove the second part.
\end{proof}
\begin{lemma} \label{lemma2}
For any function $h: \mathbb{S} \to [0,1]$, $\underset{\mathbf{s} \in \mathbb{S}}{\min} \mbox{ } H_0(h(\mathbf{s})) \ge \min( H_0(\underset{\mathbf{s} \in \mathbb{S}}{\max}\mbox{ } h(\mathbf{s})), H_0(\underset{\mathbf{s} \in \mathbb{S}}{\min}\mbox{ } h(\mathbf{s})))$.
\end{lemma}
\begin{proof}
This lemma directly follows from Lemma~\ref{lemma1}. If $h(\mathbf{s}) \le 0.5$. then using lemma~\ref{lemma1}, we have: $\underset{\mathbf{s} \in \mathbb{S}}{\min} \mbox{ } H_0(h(\mathbf{s})) \equiv H_0(\underset{\mathbf{s} \in \mathbb{S}}{\min}\mbox{ } h(\mathbf{s}))$. Similarly if $h(\mathbf{s}) \ge 0.5$, $\underset{\mathbf{s} \in \mathbb{S}}{\min} \mbox{ } H_0(h(\mathbf{s})) \equiv H_0(\underset{\mathbf{s} \in \mathbb{S}}{\max}\mbox{ } h(\mathbf{s}))$. Hence $\underset{\mathbf{s} \in \mathbb{S}}{\min} \mbox{ } H_0(h(\mathbf{s}))$ is definitely greater than the minimum of these two. Hence proved.
\end{proof}
\end{document}